\documentclass{article}
\usepackage{graphicx}
\usepackage{todonotes}
\usepackage{geometry}
\usepackage{amsthm}
\usepackage{amsmath}
\usepackage{amssymb}
\usepackage[sort]{natbib}
\usepackage{xspace}
\usepackage{soul}
\usepackage{url}

\usepackage{enumitem}
\setlist[itemize]{itemsep=4pt, nolistsep,topsep=3pt}
\setlist[enumerate]{itemsep=4pt, nolistsep,topsep=3pt}

\usepackage{authblk}

\usepackage{tikz}
\usetikzlibrary{arrows.meta, calc, decorations.markings}

\newtheorem{proposition}{Proposition}
\newtheorem{observation}{Observation}
\newtheorem{lemma}{Lemma}
\newtheorem{theorem}{Theorem}
\newtheorem{corollary}{Corollary}

\theoremstyle{remark}
\newtheorem*{remark}{Remark}

\newcommand{\kh}[1]{#1}
\newcommand{\leo}[1]{#1}
\newcommand{\leoapril}[1]{#1}
\newcommand{\leomay}[1]{#1}
\newcommand{\leojuly}[1]{#1}
\newcommand{\nh}[1]{#1}
\newcommand{\nhh}[1]{#1}
\newcommand{\vm}[1]{#1}

\usepackage[normalem]{ulem}

\newcommand{\remove}[1]{}
\newcommand{\add}[1]{#1}

\usepackage{tikz}

\newcommand{\scr}{scr\xspace}
\newcommand{\cP}{\mathcal{P}}

\title{Characterizing semi-directed phylogenetic networks and their multi-rootable variants}
\author[1]{Niels~Holtgrefe}
\author[2]{Katharina~T.~Huber}
\author[1]{Leo~van~Iersel}
\author[1]{Mark~Jones}
\author[2]{Vincent~Moulton\footnote{Corresponding author: \texttt{v.moulton@uea.ac.uk}}}%\footnote{Corresponding author: \texttt{l.j.j.vaniersel@tudelft.nl}}}

\affil[1]{Delft Institute of Applied Mathematics, Delft University of Technology, Delft, The Netherlands %Mekelweg 4, 2628CD, 
}
\affil[2]{School of Computing Sciences, University of East Anglia, Norwich, United Kingdom %NR4 7TJ,
}

\date{\small\today}

\begin{document}

\maketitle

%\todo[inline]{Target journal: Graphs and Combinatorics}

%\todo[inline]{Notation: Use~$G$ for mixed graph, $N$ for semi-directed, $D$ for directed/rooting}

\begin{abstract}
In evolutionary biology, phylogenetic networks are 
graphs that provide a flexible framework for representing complex evolutionary histories 
that involve reticulate evolutionary events. Recently phylogenetic studies have started to focus on a 
special class of such networks called {\em semi-directed networks}.
These graphs are defined as mixed graphs that can be obtained
by de-orienting some of the arcs in some {\em rooted phylogenetic network}, that is,
a directed acyclic graph whose leaves correspond to a collection
of species and that has a single source or root vertex. However, this
definition of semi-directed networks is implicit in nature since it is not clear when a mixed-graph 
enjoys this property or not. In this paper, we introduce novel, explicit mathematical characterizations of 
semi-directed networks, and also {\em multi-semi-directed networks}, that is 
mixed graphs that can be obtained from directed phylogenetic networks that may have more than one root.
In addition, through extending foundational tools from the theory of rooted networks 
into the semi-directed setting — such as cherry picking sequences, omnians, and path partitions — we characterize when a (multi-)semi-directed network can be obtained by de-orienting some 
rooted network that is contained in one of the well-known classes of 
tree-child, orchard, tree-based or forest-based networks. 
These results address structural aspects of (multi-)semi-directed networks and pave 
the way to improved theoretical and computational analyses of such networks, for example, within 
the development of algebraic evolutionary models that are based on such networks. 
%Our work contributes to the growing body of research on (multi-)semi-directed networks and offers new avenues for future studies in phylogenetic graph theory.
\end{abstract}

\section{Introduction}

Phylogenetic networks are a generalization of evolutionary trees
that are used to represent evolutionary histories of organisms such as plants and
viruses that can evolve in a non-tree-like fashion (see e.g. \cite{huson2010phylogenetic}).  
In particular, they permit the representation of {\em reticulate} events, 
in which, for example, two species cross with one another or transfer genes.
There are several classes
of phylogenetic networks,  but in this paper we will mainly focus
on rooted phylogenetic networks (see e.g. \cite{kong2022classes} for a recent overview), and some of
their recent generalizations. Essentially, a rooted
phylogenetic network is a directed acyclic graph, usually with a single source or {\em root}, whose
leaf set corresponds to some collection of species or taxa. For example, 
in Figure~\ref{fig:intro}, $D_1$ is a rooted phylogenetic network for the collection $\{x_1,\dots,x_8\}$
of species. Lately, {\em multi-rooted networks} have also become of interest\remove{(see e.g.~\cite{maxfield2025dissimilarity,scholz2019osf,soraggi2019general})}, which only differ from 
rooted phylogenetic networks in that they may have 
multiple roots (see e.g. $D_2$ in Figure~\ref{fig:intro}).
\add{Such networks can be used to  model ancestral relationships between 
populations \cite{soraggi2019general} and for representing the 
evolutionary history of distantly related groups of species that 
can still exchange genes \cite{huber2022forest,scholz2019osf}.}
In both types of networks, vertices with indegree greater 
than~$1$ are of special interest because they represent 
reticulate events. Hence, such vertices are commonly called \emph{reticulations} and 
their incoming arcs \emph{reticulation arcs}.

\begin{figure}[htb]
    \centerline{\includegraphics[scale=.8]{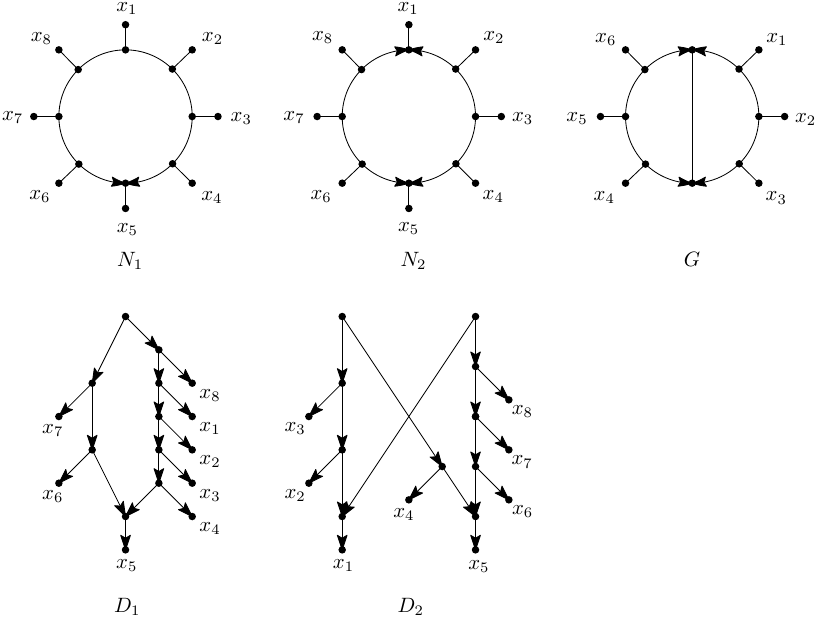}}
    \caption{\label{fig:intro} 
    Five mixed graphs~$N_1,N_2,G,D_1,D_2$. Mixed graph~$N_1$ is a semi-directed network since it is the semi-deorientation of, for example, the rooted network~$D_1$ illustrated below it. Mixed graph~$N_2$ is a multi-semi-directed network since it is the semi-deorientation of, for example, the multi-rooted network~$D_2$ illustrated below it. However, it can be shown that~$N_2$ is not a semi-directed network (since two roots are needed). 
    The mixed graph~$G$ is not semi-directed nor multi-semi-directed.
      }
\end{figure}

Recently, a new class of phylogenetic network\leojuly{s} called \emph{semi-directed networks} \cite{solis2016inferring}
has started to receive a lot of attention in the literature, both from 
a theoretical (e.g. \cite{banos2019identifying,gross2018distinguishing,gross2021distinguishing,linz2023exploring,xu2023identifiability,englander2025identifiability}) and 
applied (e.g. \cite{allman2024nanuq+,holtgrefe2025squirrel}) point of view. Basically speaking,
these are mixed graphs (i.e. graphs that can have a combination 
of undirected edges and \leojuly{directed} arcs), 
which can be obtained by partly
deorienting
%{\em deorienting}
a rooted 
phylogenetic network, that is, by replacing all arcs %in that phylogenetic 
which are not reticulation arcs
%incident with another arc having the same head
with an edge, and then suppressing any resulting degree-2 vertex that arises from a root. 
We call such a partial deorientation, in which only reticulation arcs keep their direction and the root locations are lost, a \emph{semi-deorientation}.
For example, in Figure~\ref{fig:intro}, $N_1$ is a semi-directed network since it is  
the semi-deorientation of\leomay{, e.g.,} the rooted network~$D_1$. In this paper, we will 
also consider {\em multi-semi-directed networks},  
that is, mixed graphs which
are the semi-deorientation of some 
multi-rooted network. \add{Note that these networks were recently considered in~\cite{maxfield2025dissimilarity}
in the context of defining a dissimilarity measure between semi-directed networks.}
As an example, in Figure~\ref{fig:intro}, $N_2$ is multi-semi-directed since it 
is the semi-deorientation of~$D_2$. In general, in case
a (multi-)semi-directed network $N$ is a semi-deorientation of 
a (multi-)rooted network $D$, we shall call $D$ a {\em rooting} of $N$.

\remove{Semi-directed networks have gained importance since there can be difficulties in locating the root in a phylogenetic network
when inferring networks from biological data~\cite{solis2016inferring,gross2018distinguishing,banos2019identifying,gross2021distinguishing}.} \add{Locating the root in a phylogenetic network inferred from biological data is often problematic. In particular, without asymmetrical models of character change, root placement cannot be determined from the data alone and must rely on external assumptions~\cite{kinene2016rooting}. Due to more favorable identifiability results that circumvent the need for root placement~\cite{solis2016inferring,gross2018distinguishing,banos2019identifying,gross2021distinguishing}, semi-directed networks have gained importance instead.} However, their definition
is somewhat problematic, in that it is given implicitly rather than explicitly.
For instance, the mixed graph~$G$ in Figure~\ref{fig:intro} is neither semi-directed nor 
multi-semi-directed, but how can this be decided? Of course, one possibility would be
to develop some algorithm to make this decision, 
but for certain applications it could also be useful to have 
characterizations for when a mixed graph is semi-directed or multi-semi-directed.

In this paper, we shall provide some characterizations for semi-directed \nh{and} multi-semi-directed networks. \add{These results reveal combinatorial features of (multi-)semi-directed networks that help support further results in this paper and may also be helpful for future studies.} One of the main tools that we use to obtain our characterizations
is the concept of a {\em semi-directed cycle} in a mixed graph, that is, a cycle of arcs and
edges in the graph whose edges can be oriented so as to obtain a directed cycle.
Indeed, the exclusion of such cycles is a condition in one of our main characterizations
(see e.g. Theorem~\ref{thm:semi-directed-cycle}). As a corollary of our characterizations,
in case a mixed graph is a \mbox{(multi-)semi}-directed network,
we give a more general 
characterization to that given in \cite{maxfield2025dissimilarity}
for when a subset of vertices 
or subdivisions of edges or arcs can correspond to 
a choice of root(s) that leads to a rooting of the network.

We shall also explore the consequences of our results
for some special classes of rooted phylogenetic networks
(see e.g. \cite{kong2022classes} for a recent review of 
the different types of rooted networks).  In particular, given a 
fixed class of rooted or multi-rooted phylogenetic networks, it is of interest to 
characterize when %a rooting of
a (multi-)semi-directed network \leojuly{has a rooting} contained in the given class.
This could be either strongly (all rootings are in the class) or 
weakly (there exists a rooting that is in the class). For example, a 
rooted \add{phylogenetic} network is \emph{tree-based} \cite{francis2015which} 
if it has a rooted spanning tree with the same leaf set as the network. 
Thus, a semi-directed network is \emph{weakly tree-based} if it 
has a rooting that is tree-based, 
and it is \emph{strongly tree-based} 
if all its rootings are tree-based (see e.g.  
Figure~\ref{fig:intro_tb}).

\begin{figure}[htb]
\centering
 \includegraphics{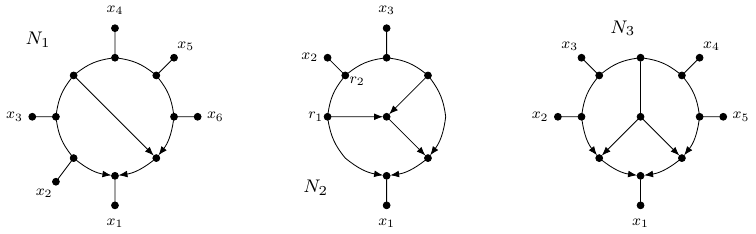}  
\caption{\label{fig:intro_tb} 
    Three semi-directed networks~$N_1$, $N_2$ and $N_3$. Semi-directed network~$N_1$ 
    is strongly tree-based since each of its rootings is tree-based. 
    % since it has only one omnian~$o$ (see Theorem~\ref{thm:strongTB}). 
    Semi-directed network~$N_2$ is weakly tree-based, 
    but not strongly, since the rooting obtained by directing all edges away from the vertex $r_1$ is tree-based, but the rooting obtained in
    a similar way using vertex $r_2$ is not.
    % (see also Figure~\ref{fig:tree-based} and the characterizations in Theorem~\ref{thm:strongTB} and Corollary~\ref{cor:weak-tree-based}). 
    Semi-directed network~$N_3$ is not weakly tree-based since it has no rooting that is tree-based.}
    % since the two vertices $v_1$ and $v_2$ can not both be covered by a collection of paths satisfying property~P2 of Corollary~\ref{cor:weak-tree-based}.}}
    %%
    %% Add the commented out remarks later in when presenting the example.
\end{figure}

%\leomay{For example, a (multi-)rooted network is \emph{tree-child} if each non-leaf vertex has at least one child that is not a reticulation. A (multi-)semi-directed network is \emph{weakly tree-child} if it has an orientation that is tree-child and it is \emph{strongly tree-child} if all its orientations are tree-child (see~\cite{maxfield2025dissimilarity,englander2025identifiability,allman2025level1identifiability}). Figure~\ref{fig:intro_tc} illustrates our characterizations of weakly and strongly tree-child. Also see~\cite[Proposition~11]{maxfield2025dissimilarity} for an alternative characterization.}

Various concepts have been used to characterize when a rooted phylogenetic 
network is contained within a certain class.\remove{For example, cherry picking sequences 
can be used for tree-child networks \cite{cardona2009comparison} and orchards \cite{erdos2019class,janssen2021cherry},
omnians \cite{jetten2016nonbinary} for tree-based networks \cite{francis2015which}, 
and path partitions for forest-based networks \cite{huber2022forest}.} \add{For example, the classes of rooted tree-child networks~\cite{cardona2009comparison} and rooted orchard networks~\cite{erdos2019class,janssen2021cherry} can both be characterized using cherry picking sequences~\cite{erdos2019class,janssen2021cherry}, eventually leading to practical software for reconciling phylogenetic trees into networks~\cite{bernardini2023constructing,bernardini2024inferring}. On the other hand, rooted tree-based networks~\cite{francis2015which} have been characterized using both omnians~\cite{jetten2016nonbinary} and path partitions~\cite{francis2018new}, with the latter concept also providing a characterization for rooted forest-based networks~\cite{huber2022forest}.}
We shall generalize some of these concepts to 
(multi-)semi-directed networks and apply them to obtain characterizations for when 
such a network has a rooting
%can be rooted\todo{again rooted not defined?} to give a network 
that is contained  within 
the classes of tree-child, orchard, tree-based or
forest-based networks (or when all of its rootings are contained within these classes).

\subsection{Previous work} 

There is a well-established body of literature devoted to the study of rooted phylogenetic 
networks, with numerous structural classes receiving extensive attention (see again \cite{kong2022classes}). 
In parallel, undirected phylogenetic networks have also been studied when no information of directionality is available (see e.g. \cite{gambette2012quartets}). 
This has led to work that focuses on the relationship between rooted and undirected networks. Recent contributions in this area include studies such as \cite{HUBER2024103480,maeda2023orienting, van2018unrooted, fischer2020tree,urata2024orientability,dackcker2024characterising, docker2025existence,dempsey2024wild,garvardt2023finding}, which explore 
how undirected networks can be oriented and when this  can be done to give rooted networks within specific classes.

In terms of (multi)-semi-directed networks, similar structural and graph-theoretical questions remain largely unexplored. 
Indeed, although seemingly related, the partial presence of directions within these networks presents fundamentally different problems. 
To date, the only substantial work in this direction is by Maxfield, Xu and Ané \cite{maxfield2025dissimilarity}. 
The main aim of their work was to introduce an efficiently computable dissimilarity metric between two tree-child 
multi-semi-directed networks. To do this they developed some results for directing mixed graphs
as multi-semi-directed networks. 
There are, however, several important differences between their framework and ours. 
Most notably, our definition of (multi-)semi-directed networks allows rootings on edges or arcs
and we enforce arcs to correspond to reticulation arcs, whereas they do not, 
and they permit parallel edges or arcs, whereas we do not.
For their type of networks they present two results that are related to ours: a 
characterization for when a subset of vertices in a mixed graph gives a rooting of the graph \cite[Proposition~8 and Remark~1]{maxfield2025dissimilarity},
and when a mixed graph can be rooted to give a tree-child network in case such a rooting exists \cite[Proposition~11]{maxfield2025dissimilarity}. We give some
more details on the relationship of these results with ours below.

%\textbf{Differences with \cite{maxfield2025dissimilarity} that need to be highlighted: Our notion of multi-semi-directed network is different from theirs because (1) we allow rootings on edges/arcs and they do not; (2) we do not allow parallel edges/arcs, and they do; (3) we do not consider networks with arcs whose head is a tree vertex to be semi-directed networks, whereas they do (they consider any graph `between' our sd-network and a d-network to be a semi-directed network); (4) we also cover (1-)semi-directed networks separately, and they do not.} 
%They have results on root-locations of such networks...., although in a different setting than ours (see also Section~\ref{sec:rootings}). As part of their main algorithm, they provide a way of checking whether a multi-semi-directed network is weakly or strongly tree-child with a single traversal of the network, which is related to our results in Section~\ref{sec:omnians}.

\subsection{Outline of the paper}

We now summarize the contents of the rest of this paper. After 
presenting some preliminaries in Section~\ref{sec:prelim}, in Section~\ref{sec:characterize} we
present some characterizations for when a mixed graph is multi-semi- and semi-directed (Theorems~\ref{thm:sink_component}~\&~\ref{thm:semi-directed-cycle} and 
Corollary~\ref{cor:semi-directed}, respectively). Then 
in Section~\ref{sec:rootings}, we characterize 
%when a subset of vertices of a multi-semi-directed network
%can be the set of roots
the feasible sets of root locations in rootings of the network
%for a rooting of the network
(Theorem~\ref{thm:allrootings}).
At the end of the section we also explain how our
characterizations lead to an efficient algorithm for deciding 
if an arbitrary mixed graph is a multi-semi-directed or semi-directed network.
In Section~\ref{sec:omnians} we use omnians to characterize when a (multi)-semi directed  
network is strongly and weakly tree-child (Theorem~\ref{thm:stronglyTC} and 
Proposition~\ref{prop:weaklyTC}, respectively) and also 
strongly tree-based (Theorem~\ref{thm:strongTB}).
%\vm{Note that in \cite[Proposition~11]{maxfield2025dissimilarity} a related result is presented for tree-child networks, which we explain in a bit more detail in Section~\ref{sec:omnians}.}
In Section~\ref{sec:cherries} we use cherry picking sequences to characterize when a (multi)-semi directed network is 
strongly or weakly orchard (Theorem~\ref{thm:orchard} and 
Theorem~\ref{thm:strongorchard}, respectively), and 
in Section~\ref{sec:paths} we use path partitions to characterize when 
the network is weakly forest-based (Theorem~\ref{thm:weaklyforestbased}) or 
weakly tree-based (Corollary~\ref{cor:weak-tree-based}). 
In Section~\ref{sec:discussion}, we conclude with some open problems.

%\textbf{Differences with \cite{maxfield2025dissimilarity} that need to be highlighted: Our notion of multi-semi-directed network is different from theirs because (1) we allow rootings on edges/arcs and they do not; (2) we do not allow parallel edges/arcs, and they do; (3) we do not consider networks with arcs whose head is a tree vertex to be semi-directed networks, whereas they do (they consider any graph `between' our sd-network and a d-network to be a semi-directed network); (4) we also cover (1-)semi-directed networks separately, and they do not.}

\section{Preliminaries}\label{sec:prelim}
%\kh{Suppose $X$ is a finite set with at least two elements.}

\subsection{Mixed graphs}
A \emph{mixed graph} is an \kh{ordered} tuple $G=(V,E,A)$ where~$V$ is a nonempty set of vertices,~$E$ is a set of undirected \emph{edges}~$\{u,v\}\subseteq V$, $u\neq v$, and~$A$ is a set of directed \emph{arcs}~$(u,v)$ with~$u,v\in V$, $u\neq v$, and such that for all arcs~$(u,v)\in A$ \kh{we have} that $\{u,v\}\notin E$ and~$(v,u)\notin A$. \leo{Note that, \kh{by definition,} parallel arcs, parallel edges or parallel edge/arc pairs are not allowed in mixed graphs.} For an arc~$(u,v)\in A$, we call~$u$ the \emph{tail} and~$v$ the \emph{head}. \leojuly{If $(u,v)\in A$, we call $u$ a \emph{parent} of $v$ and $v$ a \emph{child} of $u$.
If there is an edge~$\{u,v\}\in E$ or an arc~$(u,v)\in A$, we call $u$ and~$v$ \emph{adjacent} or {\em neighbours}.}

\kh{Suppose for the following that $G=(V,E,A)$ is a mixed graph.} \vm{For $v\in V$, the \emph{indegree}~$d_G^-(v)$ is  the number of arcs entering~$v$, the \emph{outdegree}~$d_G^+(v)$ is the number of arcs leaving~$v$, and $d_G^e(v)$ is the number of edges \kh{in $E$} incident to~$v$. In addition, 
the \emph{degree}~$d_G(v)$ is} the total number of edges and arcs \leojuly{incident} to~$v$. We will omit the subscript~$G$ when the graph is clear from \kh{the} context. \kh{We call}~$v$ a \emph{reticulation} if~$d^-(v)>1$, a \emph{leaf} if~$d(v)=d^e(v)=1$ or~$d(v)=d^-(v)=1$ and a \emph{root} if $d^+(v)=d(v)$. \kh{The set of leaves of $G$ is called the {\em leaf set} of $G$, and it is denoted by $L(G)$.} %\nh{We call $u$ a \emph{parent} of $v$ and $v$ a \emph{child} of $u$ in~$G$ if $(u,v)\in A$.} % moved up
\kh{We say that $G$} is \emph{binary} if\remove{$d(v)\in\{1,3\}$ for all~$v\in V$.}
\add{$d(v)\in\{1,2\}$ for each root~$v\in V$ and $d(v)\in\{1,3\}$ for each non-root~$v\in V$.}
%\todo{note that this is not the standard definition in the rooted case} 

A \emph{path} in $G$ is a sequence of pairwise distinct vertices~$(v_1,\ldots,v_p)$, $p \ge 1$, such that
for all $i\in\{1,\ldots ,p-1\}$ either $(v_i,v_{i+1})$ or $(v_{i+1},v_{i})$ is an arc in $A$ or $\{v_i,v_{i+1}\}$ is an edge in $E$. Such a sequence is a \emph{semi-directed path} (from~$v_1$ to~$v_p$)
if for all $i\in\{1,\ldots ,p-1\}$ either $(v_i,v_{i+1})$ is an arc in $A$ or $\{v_i,v_{i+1}\}$ is an edge in $E$. 
A \emph{$\wedge$-path} (between~$v_1$ and~$v_p$) in $G$ is a path~$(v_1,\ldots,v_i,\ldots,v_p)$ of $N$, $p \ge 1$, such that~$(v_i,\ldots ,v_1)$ and~$(v_i,\ldots ,v_p)$ are semi-directed paths, for some~$i\in\{1,\ldots ,p\}$. 
%An \emph{undirected path} of $N$ is a path traversing only edges. 
An {\em edge-path} in $G$ is a path~$(v_1,\ldots,v_p)$, such that $\{v_i,v_{i+1}\}$ is an edge in $E$, for all $i\in\{1,\ldots ,p-1\}$. We call the number of vertices on a path $P$ minus~$1$ the {\em length} of $P$ and refer to $P$ as a {\em trivial} path if
\kh{the length of $P$ is zero.}
If~$P$ is not trivial then we sometimes also say that $P$ is {\em non-trivial}. 
\kh{We say that $G$ is  {\em connected} if for any two vertices $x$ and $y$ of $G$ there is a path joining $x$ and $y$}

\kh{A \emph{cycle} in $G$ is a sequence of vertices 
$(v_1,v_2,\ldots,v_p=v_1)$, $p \ge 4$, %\todo{so we don't allow 3-cycles? We do! Note that $v_p=v_1$.}
such that~$v_i\neq v_j$ for $1\leq i<j<p$ and, for all 
$i\in\{1,\ldots ,p-1\}$, either $(v_i,v_{i+1})$ or 
$(v_{i+1},v_{i})$ is an arc in $A$ or 
$\{v_i,v_{i+1}\}$ is an edge in $E$. \leojuly{Note that, since~$p\geq 4$ and~$v_p=v_1$, a cycle contains at least three distinct vertices.}
A cycle is called  \emph{semi-directed}
if, for all $i\in\{1,\ldots ,p-1\}$, either $(v_i,v_{i+1})$ is an arc in $A$ or 
$\{v_i,v_{i+1}\}$ is an edge in $E$. A reticulation~$r$ of $N$  is a \emph{sink} of a cycle in $N$ if~$r=v_i$ and $(v_{i-1},v_i),(v_{i+1},v_i)\in A$, for some~$i\in\{1,\ldots, p-1\}$, with~$v_0=v_{p-1}$. If $G$ is connected and does not contain a cycle then we call $G$ a \emph{tree}.} \add{If additionally $E = \emptyset$, i.e., $G$ is fully directed, and $G$ has a single root, we call $G$ a \emph{rooted tree}.} \nh{We refer to Figure~\ref{fig:defs} for examples that illustrate some of the definitions in this and the next subsection.}

%\kh{Given two vertices~$u,v$ of $G$, we say that~$v$ is \emph{below}~$u$ if there exists a semi-directed path $P$ from~$u$ to~$v$ (possibly $u=v$). If, in addition, $P$ is non-trivial, we say~$v$ is \emph{strictly below}~$u$. If $v$ is (strictly) below $u$ then we say $u$ is \emph{(strictly) above $v$}. We call $u$ a \emph{parent} of $v$ in~$G$ if $(u,v)\in A$.}

\subsection{Multi-rooted and multi-semi-directed networks}
\nh{Suppose $X$ is a finite set with at least two elements.}
A \emph{multi-rooted network \kh{(on $X$)}} is a mixed graph~$(V,E,A)$ \leojuly{(with leaf set $X$)}, with~$E=\emptyset$, no directed cycles, $d(v)\neq 2$ for all non-root vertices~$v\in V$ and~$d^-(v)\in\{0,\add{1}, d(v)-1\}$ for all~$v\in V$.
%\todo{ We need to say that we need to have at least 2(?) leaves as otherwise all sorts of undesirable anomalies can occur. Not sure. Maybe it can also cause problems when we donot allow networks with a single leaf.}
A 
%\emph{$k$-rooted directed network} 
\emph{\kh{$k$-rooted network}} is a multi-rooted network with precisely~\kh{$k\geq 1$} roots. A 
%$1$-rooted directed network
\kh{$1$-rooted network}
is also called a \emph{rooted network}. \kh{A rooted network without any reticulations is called a {\em rooted phylogenetic tree}.}

\leojuly{Consider any mixed graph.} \emph{Subdividing} an edge~$\{u,v\}$ replaces the edge~$\{u,v\}$ by two edges~$\{u,w\}$ and~$\{w,v\}$ with~$w$ a new vertex. \emph{Subdividing} an arc~$(u,v)$ replaces the arc~$(u,v)$ by an edge~$\{u,w\}$ and an arc~$(w,v)$ with~$w$ a new vertex.
\emph{Suppressing} a degree-2 vertex~$w$ is defined as follows:
\begin{itemize}
    \item if~$w$ has two incident edges~$\{u,w\},\{w,v\}$, replace them by a single edge~$\{u,v\}$;
    \item if~$w$ has two incident arcs~$(u,w),(w,v)$, replace them by a\remove{sinlge} \add{single} arc~$(u,v)$;
    \item if~$w$ has an incident edge~$\{u,w\}$ and an incident arc~$(w,v)$, replace them by a single arc~$(u,v)$; 
    \item if~$w$ has an incident arc~$(u,w)$ and an incident edge~$\{w,v\}$, replace them by a single edge~$\{u,v\}$, % note that for the reticulated cherry reductions in orchard networks, it is necessary to create an edge here and not an arc
\end{itemize} and in each \leojuly{of these cases} also delete~$w$. \leomay{Note that degree-2 vertices with two incoming or two outgoing arcs are not suppressed.}

The \emph{semi-deorientation} of a multi-rooted network~\nh{$D$} is the %mixed graph~$G$ obtained from~$N$ by
\leoapril{result of} replacing each arc~$(u,v)$ \nh{of~$D$} by an edge~$\{u,v\}$ if $d_D^-(v)=1$ and \nhh{afterwards} suppressing any vertex~$\rho$ with~$d_D(\rho)=2$. \leojuly{Note that a root with two outgoing arcs in~$D$ will still be present in the semi-deorientation.}
%\todo{NH: I think we should note that by definition, this sometimes means that $\rho$ is not even suppressed (if it still has two outgoing arcs after undirecting all non-reticulation edges). Otherwise it is a bit confusing and may seem that we really suppress every root.} 
\leoapril{Also note that a semi-deorientation is not necessarily a mixed graph since suppressing roots may lead to parallel arcs.
%\kh{which are not mixed graphs in our sense.}
\leojuly{However, we will only consider mixed graphs in this paper and hence not consider cases with parallel arcs.}
}
A \emph{rooting}\footnote{\nhh{Rootings were called ``rooted partners'' in \cite{maxfield2025dissimilarity,linz2023exploring}.}} of a mixed graph~$G$ is
a multi-rooted
network~$D$ such that~$G$ is the semi-deorientation of~$D$.
Observe that~$D$ can be obtained from~$G$ by subdividing (zero or more) arcs and/or edges and replacing edges by arcs. Note that\remove{this} \add{the %subdivided arcs and/or edges 
subdivision vertices necessarily become roots in~$D$ and that vertices of~$G$ may also become roots in~$D$. This also includes the possibility that a  root of~$G$ that is contained in~$X$ becomes a leaf in~$D$ (see e.g. the vertex $e$ in Figure~\ref{fig:defs}) or that a leaf of~$G$ not contained in~$X$ becomes an outdegree-1 root in~$D$}\footnote{\add{From a biological perspective, turning a leaf in a (multi-)semi-directed network into a root of a (multi)-rooted network may seem problematic. However, since this only applies to leaves of a (multi-)semi-directed network not in~$X$, which do not correspond to taxa, this does not pose a problem. 
A leaf of a (multi-)semi-directed network contained in~$X$ cannot be turned into a root of a (multi)-rooted network. Instead, its incident edge can be subdivided by a new vertex that becomes a root.
%If instead the leaf is in~$X$, this is handled by subdividing the edge incident to the root, thereby introducing a new root node without altering the original taxon.
}}. \remove{includes the possibility that one or more leaves become roots.} However, it is not possible to create any new reticulations.
%\todo{\kh{The terminology introduced above for multiple rooted networks used in this paper was agreed on 27/03/2025. The remainder of the paper has to be made consistent with it.}}
A %\emph{$k$-rooted semi-directed network} 
\emph{\kh{$k$-semi-directed network}} \leo{(on~$X$)} is a mixed graph that is the semi-deorientation of some $k$-rooted network (with leaf set~$X$). Note that the leaf set of a $k$-semi-directed network on~$X$ may not be equal to~$X$, because an outdegree-1 root may also be in~$X$ \add{(see, for example, the vertex~$e$ in Figure~\ref{fig:defs} again)}. 
A \emph{multi-semi-directed network} is a 
%$k$-rooted 
\kh{$k$}-semi-directed network for \kh{some~$k\geq 1$}.
A \emph{semi-directed network} is a $1$-semi-directed network. When drawing multi-rooted and multi-semi-directed networks, we will often omit the leaf labels when they are not relevant. \nh{We will reserve the letter~$G$ for general mixed graphs, $D$ for (multi)-rooted networks and $N$ for (multi)-semi-directed networks.}

%\kh{The \emph{reticulation number} of $N$} is $|A|-|R|$ with~$R$ the set of reticulations. 
%and~$A$ the set of arcs.
%\todo[inline]{ \kh{All the definitions that I have added in blue and were also used in the Quarnet paper are cut-and-paste from that paper.}}   

\section{Characterizations of multi-semi-directed networks}\label{sec:characterize}

\nh{In this section we characterize when a mixed graph is a multi-semi-directed network or when it is \kh{in fact} a semi-directed network.} \add{These characterizations naturally lead to algorithms for checking (multi-)semi-directedness. We sketch one efficient approach at the end of Section~\ref{sec:rootings}, which outlines a method based on rootings of (multi-)semi-directed networks.}

Suppose \nh{$G=(V,E,A)$} %N=(V,E,A)$ 
is a mixed graph.
\leojuly{A \emph{cherry} of a mixed graph is an ordered pair of leaves~$(x,y)$ such that there is a length-$2$ path between~$x$ and~$y$, either consisting of two edges or of two arcs directed towards~$x$ and~$y$ respectively.} 
% for the cherry picking section
\kh{For convenience, we also refer to a cherry with leaves $x$ and $y$ as a cherry on $\{x,y\}$.} %and denote it by $(x,y)$.}
A \emph{leaf reticulation} \nh{of $G$} is a reticulation \nh{of $G$} that is adjacent to a leaf \nh{of $G$} and a \emph{reticulation leaf} \nh{of $G$} is a leaf \nh{of $G$} that is adjacent to a reticulation \nh{of $G$}.

An \emph{(undirected) sink component} of $G=(V,E,A)$ is \kh{a connected component $C$ of 
the (undirected) graph $(V,E)$
such that there are no arcs~$(u,v)\in A$ with~$u\in C$ and~$v\notin C$. 
For example, in Figure~\ref{fig:defs}, $\{d,r_1\}$ is a sink component since it is a connected component of the graph 
that is obtained by ignoring all arcs, and in~$G$ there are no arcs leaving this component.
Similarly, an \emph{(undirected) source component}\footnote{Source components were called ``root components'' in~\cite{maxfield2025dissimilarity}.} of~$G$ is a connected component $C$ of %the graph obtained from~$N$ by deleting all arcs,
$(V,E)$ such that there are no arcs~$(u,v)\in A$ with~$u\notin C$ and~$v\in C$.}
%\todo{NH: In \cite{maxfield2025dissimilarity} they also talk about `root components' which seem very similar, if not the same...}
%
%\todo{should this be called an ``undirected sink component'' and an "undirected source component".} 
%A connected component~$C$ of a subgraph~$G'$ of~$G$ obtained from~$G$ by deleting a single cut-edge of $G$ is called a {\em pendant subtree} of $G$ if $C$ is a %%phylogenetic
%tree.
\leojuly{A subgraph~$G'$ of~$G$ is a \emph{pendant subtree} if~$G'$ is a tree and has at most one vertex that has a neighbour in~$G$ not in~$G'$.}
%\todo{You are correct and I have implemented the change. Please remove this box when you have read it.}

\begin{figure}[htb]
    \centerline{\includegraphics[scale=1]{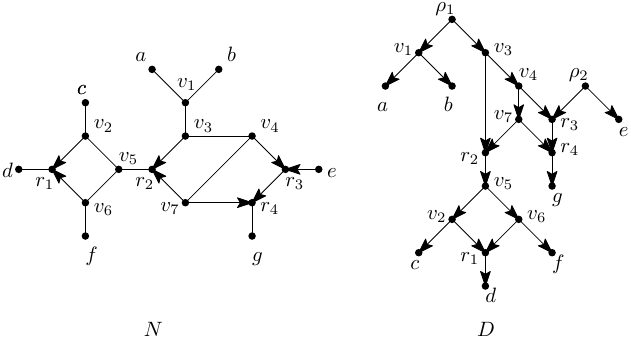}}
    \caption{\label{fig:defs} 
 \emph{Left:} A $2$-semi-directed network $N$ \kh{on $\{a,b,\ldots, e\}$} with set of reticulations~$\{r_1,r_2,r_3,r_4\}$.
 The sequence $(r_1,v_2,v_5,r_2,v_3,v_4,r_3,r_4,g)$ is a $\wedge$-path of $N$ and the sequence $C=(r_2,v_3,v_4,r_3,r_4,v_7,\kh{r_2}
 %v_2
 )$ is a cycle of $N$. The vertices $r_2$ and $r_4$ are sinks of $C$ whereas $r_3$ is not. The path $(a,v_1,v_3,v_4,v_7)$ is an example of an edge-path. The leaves~$a$ and~$b$ form a cherry, leaves~$d$ and~$g$ are reticulation leaves, while~$r_1$ and~$r_4$ are leaf reticulations. The vertex sets of the source components are $\{a,b,v_1,v_3,v_4,v_7\}$ and~$\{e\}$ while the sink components have vertex sets~$\{d,r_1\}$ and~$\{g,r_4\}$. \emph{Right:} A rooting~$D$ of~$N$ in the form of a $2$-rooted network with roots~$\rho_1$ and~$\rho_2$.}
\end{figure}
\kh{We start with the following technical result} \nh{which we use to prove our first characterization of multi-semi-directed networks.}

\begin{lemma}\label{lem:help}
Suppose $G=(V,E,A)$ is a connected mixed graph \leojuly{with~$|V|\geq 3$}. Then $G$ contains either a cherry or a leaf reticulation (or both) if the following properties hold:
\begin{enumerate}
    \item[(Ci)] $d(v)\neq 2$ and $d^-(v)\in\{0,d(v)-1\}$ for all~$v\in V$;
    \item[(Cii)] each cycle of~$G$ contains at least one sink; and
    \item[(Ciii)] each sink component of~$G$ is a pendant subtree.
\end{enumerate}
\end{lemma}

\begin{proof}
\kh{Let} $P$ be a semi-directed path in~$G$ containing a maximum number of arcs and, over all such paths, containing a maximum number of edges. If~$P$ contains no arc then it follows that~$A=\emptyset$ and hence~$G$ is a tree by (Cii) \kh{and since $G$ is connected}. \leojuly{Since~$G$ has at least three vertices, and no degree-$2$ vertices, it follows that~$G$ has a cherry.}
%In this case,~$G$ is the semi-deorientation of a rooted \kh{phylogenetic} tree and hence~$G$ is a \kh{1}-semi-directed network. 
Therefore, we may assume that~$P$ contains at least one arc. \kh{Let $s$ denote the first vertex on $P$.}

\kh{We next show} that~$P$ ends in a leaf or in a reticulation. To see this, assume \kh{for contradiction that}~$P$ ends in a vertex~$v$ with~$d(v)\neq 1$ and~$d^-(v)=0$. Then there is an edge~$\{v,w\}$ or an arc~$(v,w)$ with~$w$ not on~$P$ (since otherwise we would have a cycle without sink, which is not allowed by (Cii)). Hence, \kh{we can} extend~$P$ to a semi-directed path containing more arcs or the same number of arcs and more edges, by appending~$w$ \kh{in contradiction to the maximality of $P$. Thus,} $P$ ends in a leaf or in a reticulation, \kh{as required}.

\kh{To complete the proof, we next show} that~$P$ ends in a reticulation leaf or in a leaf that is in a cherry. To prove this, assume that~$P$ ends in a reticulation or in a leaf that is not a reticulation leaf and not in a cherry. First suppose that~$P$ ends in a leaf~$x$ that is not a reticulation leaf. 
\kh{Then, in view of (Ci), there exists a vertex $z$ of $N$ that is not on 
$P$ and, denoting by $w$ the vertex on $P$ that is the predecessor of $v$ on $P$, we have that $\{w,z\}\in E$ or $(w,z)\in A$. Hence, $z$}
%Then the sibling of~$x$ (i.e. the vertex of distance~2 from~$x$ that is not on~$P$) 
is either also a leaf or it is a reticulation. 
%(otherwise we could extend~$P$ leading to a contradiction).
In the first case,~$x$ is in a cherry. Therefore, we may and will assume that~$P$ ends in a reticulation~$v$. Furthermore, \kh{there exist} arcs entering~$v$ \kh{that} 
are not on~$P$ and, \kh{by the maximality of $P$,} the last edge/arc of~$P$ is an edge.

Consider the maximal connected subgraph~$H$ of~$G$ containing~$v$ but no arcs. Observe that~$H$ is not a pendant subtree since it contains at least two vertices with incoming arcs ($v$ and the head of the last arc of~$P$). Hence, by (Ciii), $H$ is not a sink component. It follows that~$H$ has a vertex~$y$ with an incident outgoing arc~$a=(y,w)$, some $w\in V$. \kh{Hence, the subpath of $P$ from $s$ to $y$ extended by $w$ has one more arc than $P$ contradicting }
%However, that means we can extend~$P$ by a path from~$v$ to~$y$ through~$H$ (containing only edges) and arc~$a$. This contradicts 
the assumption that~$P$ contains a maximum number of arcs. 
%We conclude that~$G$ contains a cherry or a leaf reticulation.
\end{proof}

\begin{theorem}\label{thm:sink_component}
A mixed graph~$G=(V,E,A)$
is a \kh{multi}-semi-directed network if and only if Properties (Ci) - (Ciii) hold.
\end{theorem}
\begin{proof}
If~$G$ is a \kh{multi}-semi-directed network, then  \kh{Properties~(Ci) - (Ciii)} clearly hold.

\kh{To see the converse, assume that Properties~(Ci) - (Ciii) hold. We perform induction on $|A|+|E|$. The base case is $|A|+|E|=0$. In this case, $G$ is a set of isolated} vertices. \kh{Hence, $G$ is} its own semi-deorientation. Thus,~$G$ is a multi-semi-directed network.

\kh{For the inductive step, assume that $G$ is such that $|A|+|E|\geq 1$. Then $G$ must contain a connected component that is not an isolated vertex. In view of the base case
%and because $G$ has at least 2 leaves,
% Leo: we never assumed this
it
suffices to show that every connected component of $G$ that is not an isolated vertex is a multi-semi-directed network. Let $G'$ be a connected component of $G$ that is not an isolated vertex. \leojuly{If~$G'$ has exactly two vertices, then they are connected by an edge by~(Ci) and clearly~$G'$ is (multi-)semi-directed.} Otherwise, by Lemma~\ref{lem:help}, $G'$ must contains either a cherry or a leaf reticulation (or both).}

\kh{Assume first that~$G'$ contains a cherry on $\{x,y\}$. \leojuly{Then, by~(Ci), the length-$2$ path between~$x$ and~$y$ consists of two edges.} In this case, we} delete leaf~$x$ from this cherry and suppress 
\kh{the vertex adjacent to $x$ if this has rendered it a vertex of degree two.}
%its former neighbour if it gets degree~2. 
The resulting graph~$G''$ is a multi-semi-directed network by induction. This means that~$G''$ is the semi-deorientation of a multi-rooted network~\nh{$D''$}. Let~\nh{$D'$} be obtained from~\nh{$D''$} by subdividing the arc entering~$y$ by a new vertex~$p$ and adding leaf~$x$ with an arc~$(p,x)$. Then~\nh{$D'$} is a multi-rooted network and the semi-deorientation of~\nh{$D'$} is~$G'$. Hence,~$G'$ is a multi-semi-directed network.

\kh{To conclude the proof, assume that}
%The second and last case of the induction step is that~
$G'$ contains a reticulation leaf~$z$. Let~$r$ be the reticulation adjacent to~$z$ and let~$p_1,\ldots ,p_{t}$ be the parents of~$r$ (with~$t=d^-(r)$). Let~$G''$ be obtained from~$G'$ by replacing vertices~$r,z$ by vertices~$z_1,\ldots , z_t$ and replacing arc $(p_i,r)$ by an edge $\{p_i,z_i\}$ for $i\in\{1,\ldots ,t\}$ and deleting edge~$\{r,z\}$ thus reducing~$|A|+|E|$ by~$1$. The resulting graph~$G''$ is a multi-semi-directed network by induction. This means that~$G''$ is the semi-deorientation of a multi-rooted network~\nh{$D''$}. Let~\nh{$D'$} be obtained from~\nh{$D''$} by merging~$z_1,\ldots ,z_t$ into a single vertex~$r$ and adding a leaf~$z$ with an arc~$(r,z)$. Then,~\nh{$D'$} is a multi-rooted network and the semi-deorientation of~\nh{$D'$} is~$G'$. Hence,~$G'$ is a multi-semi-directed network.
\end{proof}

The following lemma can be used to show an alternative characterization of multi-semi-directed networks.

\begin{lemma}\label{lem:cycle} Let~$C$ be a cycle in a mixed graph \nh{$G$}. If
\begin{enumerate}
    \item[(a)] $C$ is not semi-directed; and
    \item[(b)] $C$ contains no non-trivial 
    edge-path
    %\todo[inline]{\kh{do we mean edge-path because we have not defined "undirected path"?}} 
    between two reticulations \nh{of $G$},
\end{enumerate}
then $C$ contains at least one sink.
\end{lemma}
\begin{proof}
Let~$C=(v_1,v_2,\ldots ,v_p=v_1)$, \kh{$p\geq 4$}. If~$C$ is not semi-directed then there 
\kh{exist $i,j\in\{1,2,\ldots, p-1\}$ distinct such that
$(v_i,v_{i+1})$ and $(v_{j+1},v_j)$ are arcs on $C$}. Without loss of generality, $i<j$ \kh{$(\mathrm{mod}\,\,\,\,p-1)$}. \kh{Assume that $j$ and $i$ are} such that $j-i$ is minimized. If $j-i=1$ then $C$ contains \kh{the} sink $v_j=v_{i+1}$. \kh{Otherwise}, $j-i>1$ \kh{and, so, $v_{i+1}$ and $v_j$ are both reticulations of \nh{$G$} and $C$ contains a non-trivial 
%undirected path 
edge-path between them}.
\end{proof}

An alternative characterization \kh{of mixed graphs that are multi-semi-directed networks} is as follows. See Figure~\ref{fig:semi-directed-cycle}
\kh{for examples that illustrate that Properties~(II) and (III) cannot be weakened}.

\begin{figure}[htb]
    \centering
    \includegraphics{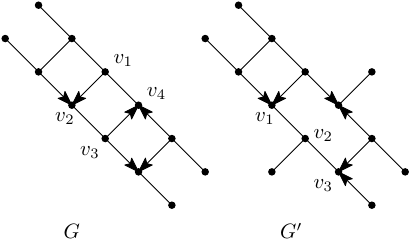}
    \caption{\label{fig:semi-directed-cycle} Two mixed graphs \kh{$G$ and $G'$} that, \kh{by Theorem~\ref{thm:semi-directed-cycle},} are not \kh{multi}-semi-directed networks. \kh{The reason is that}~$G$ contains the semi-directed cycle $(v_1,v_2,v_3,v_4,v_1)$ while~$G'$ contains \kh{the non-trivial edge-path}~$(v_1,v_2,v_3)$ \kh{and $v_1$ and $v_3$ are} reticulations.}
\end{figure}

\begin{theorem}\label{thm:semi-directed-cycle}
A mixed graph~$G=(V,E,A)$ is a \kh{multi}-semi-directed network if and only if
\begin{enumerate}
    \item[(I)] $d(v)\neq 2$ and $d^-(v)\in\{0,d(v)-1\}$ for all~$v\in V$;
    \item[(II)] $G$ contains no semi-directed cycle; and
    \item[(III)] $G$ contains no non-trivial \leo{edge-}path between two reticulations.
\end{enumerate}
\end{theorem}
\begin{proof}
\kh{As in the case of Theorem~\ref{thm:sink_component}, it is straight-forward} to see that if~$G$ is a \kh{multi-}semi-directed network, then \kh{Properties}~(I), (II) and (III) hold.

For the \kh{converse} direction, by Theorem~\ref{thm:sink_component}, \kh{it suffices} to show that \kh{Properties}~(I), (II) and (III) imply \kh{Properties}~(Ci), (Cii) and (Ciii). \kh{For this, it suffices to show that Properties~(II) and (III) imply Properties~(Cii) and (Ciii).}

First observe that \kh{Property}~(III) implies \kh{Property}~(Ciii) because any sink component that is not a pendant subtree contains two reticulations and a non-trivial edge-path path \kh{joining} them.  \kh{In view of Lemma~\ref{lem:cycle},} \kh{Properties}~(II) and (III) together imply \kh{Property}~(Cii).
\end{proof}

We will now characterize semi-directed networks. See Figure~\ref{fig:multi-semi-directed} \kh{for examples that illustrate that Properties~(2) and (3) cannot be weakened}. \nh{Recall that a semi-directed network is the semi-deorientation of some 1-rooted network, i.e. a rooted network with a single root.}

\begin{figure}[htb]
    \centering
    \includegraphics{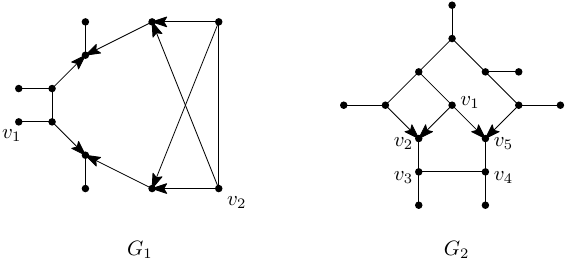}
    \caption{\label{fig:multi-semi-directed} \emph{Left:} \kh{A mixed graph $G_1$ which, by Corollary~\ref{cor:semi-directed},} is not a semi-directed network since there is no $\wedge$-path between~$v_1$ and~$v_2$. \emph{Right:} A mixed graph~$G_2$ that, \kh{by Corollary~\ref{cor:semi-directed},} is not a semi-directed network  since it contains a cycle $(v_1,v_2,v_3,v_4,v_5,v_1)$ without a sink.}
\end{figure}

\begin{corollary}\label{cor:semi-directed}
A mixed graph~$G=(V,E,A)$ is a semi-directed network if and only if
\begin{enumerate}
    \item[(1)] $d(v)\neq 2$ and $d^-(v)\in\{0,d(v)-1\}$ for all~$v\in V$;
    \item[(2)] each cycle of~$G$ contains at least one sink; and
    \item[(3)] $G$ contains a $\wedge$-path between each pair of vertices~$u,v\in V$.
\end{enumerate}
\end{corollary}
\begin{proof}
As before, it is easy to see that if~$G$ is a semi-directed network then \kh{Properties}~(1), (2) and (3) hold.

 Now assume that \kh{Properties}~(1), (2) and (3) hold. \kh{We first show that $G$ is a multi-semi-directed network. In view of Theorem~\ref{thm:sink_component},
 it suffices to show that Properties~(1) - (3) imply Property~(Ciii). To this end,}  suppose that~$G$ is contained a sink component~$S$ that is not a pendant subtree. Note that~$S$ is a tree by \kh{Property}~(2). Since~$S$ is not pendant, it contains at least two reticulations~$r_1,r_2$. Let~$p_1,p_2$ be parents of~$r_1,r_2$, respectively, such that~$p_1\neq p_2$. Note that~$p_1,p_2$ are not in~$S$. By \kh{Property}~(3), there exists a $\wedge$-path~$P$ between~$p_1$ and~$p_2$. Hence, there is a cycle without sink, formed by path~$P$ together with the arcs~$(p_1,r_1),(p_2,r_2)$ and the path between~$r_1$ and~$r_2$ through~$S$, which is a contradiction to (2). Thus,  (Ciii) holds, \kh{as required}. 
 %Therefore, by Theorem~\ref{thm:sink_component}, $G$ is a multi-semi-directed network.

It remains to show that any rooting of~$G$ has a single root. Suppose~$D$ is a rooting of~$G$ with at least two roots $\rho_1,\rho_2$. Then there is no $\wedge$-path between~$\rho_1$ and~$\rho_2$ in~$D$ and hence also not in~$G$, contradicting \kh{Property}~(3). Hence, we can conclude that~$G$ is a semi-directed network.
\end{proof}

\section{Rootings of multi-semi-directed networks}\label{sec:rootings}

The following theorem and its corollary show that the number of roots of a rooting of a \kh{multi-}semi-directed network~$N=(V,E,A)$ can be calculated directly from the \emph{reticulation number} $|A|-|R|$ \kh{of $N$}, with~$R$ the set of reticulations.

\begin{theorem}
\label{thm:reticulationsroots}
If~\nh{$N=(V,E,A)$} is a \kh{k-}semi-directed network, \kh{some $k\geq 1$,}  then the reticulation number of~$N$ equals
\[
|E|+|A|-|V|+k.
\]
\end{theorem}
\begin{proof}

Consider the graph~$F$ obtained from~\nh{$N$} by deleting all arcs. Then~$F$ is a forest since, by Theorem~\ref{thm:semi-directed-cycle}, \nh{$N$} does not contain any cycles traversing only edges. In any rooting~\nh{$D$} of~\nh{$N$}, \nh{where we root only on vertices of~$N$}, each \kh{(connected)} component~$T$ of~$F$ is oriented as a rooted tree~$T^r$. The root of~$T^r$ is either a root of~\nh{$D$} or a reticulation \nh{of $D$}. Moreover, each vertex of~$T^r$ that is not the root of~$T^r$ is not a reticulation \nh{of $D$} and also not a root of~\nh{$D$}. Hence, each component of~$F$ contains exactly one vertex that is a root or a reticulation of~\nh{$D$}. It follows that the number of components of~$F$ is precisely $|R|+k$ \kh{where $R$ is the set of reticulations of $N$.}

In any forest, the number of components equals the number of vertices minus the number of edges. Therefore,
\[
|V|-|E| = |R|+k,
\]
implying that
\[
|A| + |V|-|E| = |A|+|R|+k,
\]
which can be rewritten as
\[
|A| - |R| = |E| + |A| -|V|+ k.
\]
Since~$|A|-|R|$ is the reticulation number \nh{of $N$}, the \kh{theorem} follows. 
\end{proof}

\kh{As an immediate consequence, we have the following result.}

\begin{corollary}\label{cor:roots}
    If~$N=(V,E,A)$ is a multi-semi-directed network, then all rootings of~\nh{$N$} have 
    $$|V|-|R|-|E|$$ 
    roots, with~$R$ the set of reticulations of~\nh{$N$}.
\end{corollary}

\nh{We now prove an auxiliary result that will be useful to characterize all possible rootings} \vm{of a multi-semi-directed network.}

\begin{lemma}\label{lem:rootings}
Let~$N$ be a multi-semi-directed network, $R$ the set of reticulations of~$N$, and~$U\subseteq V(N)$ such that there is no edge-path between any two vertices in~$U\cup R$. Then there exists a rooting of $N$ in which each vertex of $U$ is a root.
\end{lemma}
\begin{proof} The proof is by induction on~$|U|$. The base case for~$U=\emptyset$ is trivial. If~$|U|\geq 1$, consider a vertex~$u\in U$. By induction, there exists a rooting~$D$ of~$N$ in which each vertex of~$U\setminus\{u\}$ is a root. Let~$C_u$ be the connected component of the graph obtained from~$N$ by deleting all arcs, such that~$C_u$ contains~$u$. First observe that~$u$ is the only vertex from~$U$ in~$C_u$ since otherwise there would be an edge-path between two vertices in~$U$. Also note that~$C_u$ is a source component since there is no edge-path between~$u$ and any reticulation. Moreover, by Theorem~\ref{thm:semi-directed-cycle}, $N$ contains no cycle consisting of only edges. Hence,~$C_u$ is a tree. Since \kh{no vertex of~$C_u$ is a reticulation in $C_u$}, it follows that~$C_u$ is a rooted tree in~$D$. The root~$u'$ of this tree is also a root of~$D$, since~$D$ contains no arcs 
%entering~$C_u$
\kh{whose head is a vertex in~$C_u$ but whose tail is not}. Hence, in~$D$, $C_u$ contains exactly one root~$u'$. 
\kh{If $u'=u$ then $u'\in U$ and $D$ is a rooting of $N$ in which $u'$ is a root. Hence, the lemma holds in this case.}

If~$u'\neq u$ \kh{then, since $C_u$ is a rooted tree in $D$,} we can modify the orientation of~$D$ to obtain an alternative rooting of~$N$ by changing \kh{only} the orientation of each arc on the path between~$u$ and~$u'$ in~$C_u$. This gives a rooting~$D'$ of~$N$ in which~$u\in U$ is a root. Moreover, all vertices in~$U\setminus\{u\}$ are still roots of~$D'$, which concludes the proof.
\end{proof}

For a rooting~\nh{$D$} of multi-semi-directed network~$N=(V,E,A)$, define the\remove{\emph{root locations}} \add{\emph{root configuration}}
%in~$N$ for~$N_d$
as the triple $(V',E',A')$ with $V'\subseteq V$, $E'\subseteq E$ and~$A'\subseteq A$ such that the roots of~\nh{$D$} are precisely the vertices in~$V'$ together with vertices subdividing each edge in~$E'$ and arc in~$A'$. In the following theorem, we characterize the valid\remove{root locations} \add{root configurations} of a multi-semi-directed network. 
As noted in the introduction, a similar characterization was given in~\cite[Proposition~8]{maxfield2025dissimilarity}, although %it was for the special case where all roots of the multi-rooted network are vertices of the multi-semi-directed network
\leojuly{it was under a slightly different framework (e.g. assuming $A' = E' = \emptyset$).}
% A further difference is that, in \cite{maxfield2025dissimilarity}, a different definition of multi-semi-directed networks is used in which parallel edges/arcs are allowed and where arcs do not necessarily need to enter a reticulation.}

%\leomay{Note that a similar characterization was given in~\cite[Proposition~8]{maxfield2025dissimilarity} for the special case that all roots of the multi-rooted networks are vertices of the multi-semi-directed network.}

%\todo[inline]{The following result is essentially proved as Prop.~8 in \cite{maxfield2025dissimilarity} in the case $E' = A' = \emptyset$, i.e. when we are only allowed to root on vertices. In the prove of our theorem we also only prove this case and argue that the cases where $E' = A' \neq \emptyset$ are similar...}

\begin{theorem}\label{thm:allrootings}
Let~$N=(V,E,A)$ be a multi-semi-directed network and~$V'\subseteq V$, $E'\subseteq E$ and~$A'\subseteq A$. Then there exists a rooting~\nh{$D$} of~$N$ with\remove{root locations} \add{root configuration} $(V',E',A')$ if and only if \begin{itemize}
    \item each vertex in~$V'$ and each edge in~$E'$ is in a source component of~$N$ and each arc in~$A'$ is an outgoing arc of a source component of~$N$; and
    \item each source component of~$N$ contains exactly one element of~$V'$, contains exactly one edge in~$E'$ or has exactly one outgoing arc in~$A'$.
\end{itemize}
\end{theorem}
\begin{proof}
    We first prove the ``only if'' direction. Since~$N$ is multi-semi-directed, there exists a rooting~\nh{$D$} of~$N$. Let~$(V'',E'',A'')$ be its\remove{root locations} \add{root configuration}.
    
    First suppose that there is a~$v\in V''$ that is not in a source component of~$N$. Then there exists a reticulation~$r$ in the component~$C$ of the graph $(V,E)$ containing~$v$ with both incoming arcs of~$r$ having their tail outside~$C$. Note that~$v\neq r$. Hence, there exists an edge-path between~$r$ and~$v$ in~$N$. Since~$v$ is a root of~\nh{$D$}, this path is a directed path from~$v$ to~$r$ in~\nh{$D$}. However, this implies that %$d^-(r)=d(r)$
    \nh{$d_D^-(r)=d_D(r)$},  contradicting the definition of a multi-rooted network. The other cases (that there is a $e\in E''$ that is not in a source component or an arc~$a\in A''$ that is not an outgoing arc of a source component) are handled similarly.

    Now suppose that there exists a source component~$C$ of~$N$ containing~$v,w\in V''$ with~$v\neq w$. Then there exists an edge-path between~$v$ and~$w$ of length at least~$2$ in~$N$. Since~$v$ is a root in~\nh{$D$}, this path is directed from~$v$ to~$w$ in~\nh{$D$}. However, since~$w$ is also a root in~\nh{$D$}, the path is directed from~$w$ to~$v$ in~\nh{$D$}, a contradiction. The other cases are again handled similarly. 

    We now prove the ``if'' direction. Since~$N$ is multi-semi-directed, there exists a rooting~\nh{$D$} of~$N$. Let~$(V'',E'',A'')$ be its\remove{root locations} \add{root configuration}.
    %and call them real root locations. Call the elements of~$V',E',A'$ the desired root locations.
    For an edge~$\{u,v\}$ of~$N$, clearly, there exists a rooting of~$N$ with root~$u$ if and only if there exists a rooting with root~$w$ subdividing~$\{u,v\}$. Similarly, for an arc~$(u,v)$ of~$N$, clearly, there exists a rooting of~$N$ with root~$u$ if and only if there exists a rooting with root~$w$ subdividing~$(u,v)$. Hence, we may assume $E'=A'=E''=A''=\emptyset$. Let~$R$ be the set of reticulations of~$N$. Since each element of~$V'$ is in a different source component, there is no edge-path between any two vertices in~$R\cup V'$.
    Then, by Lemma~\ref{lem:rootings}, it follows that there exists a rooting of~$N$ with\remove{root locations} \add{root configuration}~$(V',E',A')$.  
\end{proof}

Theorem~\ref{thm:allrootings} directly leads to an efficient algorithm \add{taking $O(|V| + |E| + |A|)$ time} for deciding if an arbitrary mixed graph~\remove{$G$}\add{$G = (V, E, A)$} is a (multi-)semi-directed network
(note that in \cite[Remark~1]{maxfield2025dissimilarity} a similar algorithm is sketched to efficiently find a rooting of a multi-semi-directed network).
First, find all the source components of~$G$ \add{by traversing the graph in $O(|V| + |E| + |A|)$ time}. Second, in each source component, pick an arbitrary vertex as root. Third, do a breath-first search from 
each root and orient all edges away from the root\add{, again taking $O(|V| + |E| + |A|)$ time}. Finally, \add{with yet another traversal,} check if the resulting mixed graph~$G'$ is a multi-rooted directed network 
and whether its semi-deorientation is~$G$. Note that by additionally checking whether or not~$G'$ has a single root, it can also be determined if~$G$ is semi-directed.

%If you want to know whether~$G$ is semi-directed, additionally check whether~$G'$ has a single root.

%\todo[inline]{In Remark~1 of \cite{maxfield2025dissimilarity} they describe a similar algorithm to efficiently find a rooting of a multi-semi-directed network. We use a similar idea to also check if any mixed graph (not necessarily multi-semi-directed) is (multi-)semi-directed.}

\section{Omnians}\label{sec:omnians}

In this section, we show how the concept of omnians can 
be used to characterize tree-child (Theorem~\ref{thm:stronglyTC} and  Proposition~\ref{prop:weaklyTC}) 
and tree-based multi-semi-directed networks (Theorem~\ref{thm:strongTB}).
As remarked in the introduction, \cite[Proposition~11]{maxfield2025dissimilarity} 
gives an alternative characterization for the tree-child case where, instead of using omnians %they 
the concept of the  ``directed part'' of a network \kh{is used}.
%\kh{(Theorem~\ref{thm:strongTB})}. 
%\nh{See also \cite{allman2025level1identifiability,englander2025identifiability} for recent use of tree-child semi-directed networks from an identifiability perspective}
%\kh{Suppose that $N$ is a mixed graph. If $N$ is a  multi-rooted network,

\leojuly{The main definitions of this section are the following. We say that a vertex~$v$ of a multi-semi-directed network is an \emph{omnian} if $d^+(v)\geq 1$ and $d^e(v)\leq 1$, see Figure~\ref{fig:omnians}.}
\nh{We call a multi-rooted network $D$} \emph{tree-child} if each non-leaf vertex \nh{of $D$} has at least one child that is not a reticulation.
%\kh{If $N$ is a multi-semi-directed network, then we call a vertex $v$ of $N$ an \emph{omnian} of $N$ if $d^+(v)\geq 1$ and $d^e(v)\leq 1$.
Furthermore, we say that a multi-semi-directed network $N$ is \emph{weakly tree-child} if~$N$ has a rooting that is tree-child. \kh{In this case, we also call such a rooting of $N$ a {\em tree-child rooting of~$N$.}}  \kh{Finally, we say} that~$N$ is \emph{strongly tree-child} if every rooting of~$N$ is tree-child. See Figure~\ref{fig:tree-child} for examples.

\begin{figure}
    \centering
    \includegraphics[]{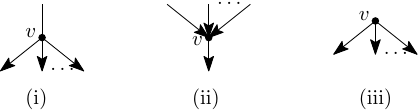}
    \caption{\label{fig:omnians} \leojuly{The three possible\remove{shapes of} \add{configurations surrounding} an omnian in a multi-semi-directed network: (i) one incident edge and at least two outgoing arcs, (ii) at least two incoming arcs and one outgoing arc and (iii) at least two outgoing arcs (and in all cases no other incident edge/arcs). Note that Case (iii) is possible in multi-semi-directed networks but not in semi-directed networks.}}
\end{figure}

\begin{figure}[htb]
    \centerline{\includegraphics[scale=.8]{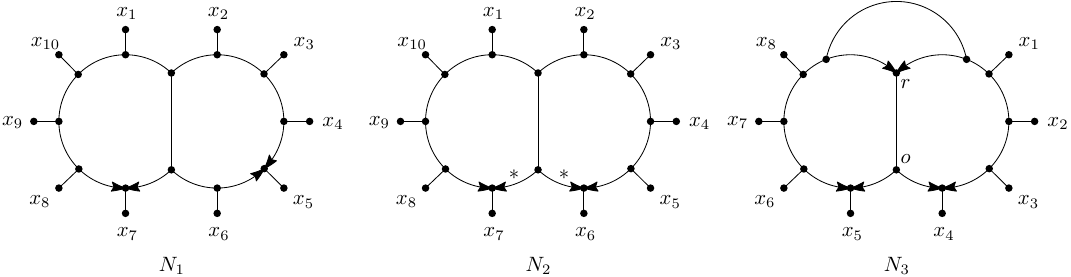}}
    \caption{\label{fig:tree-child} 
    \nh{\kh{The} semi-directed network~$N_1$ \kh{on $X=\{x_1,\ldots, x_{10}\}$}
    is strongly tree-child \kh{in view of Theorem~\ref{thm:stronglyTC}}. \kh{The} semi-directed network~$N_2$ \kh{also on $X$} is weakly tree-child, 
    but not strongly, since the rootings with the root subdividing an arc labelled $*$ are tree-child, but all other rootings are not tree-child (see Corollary~\ref{cor:wtc_sdn}). \kh{The} semi-directed network~$N_3$  \kh{on $\{x_1\ldots, x_8 \}$} is not weakly tree-child since the edge~$\{r,o\}$ forms an edge-path between a reticulation~ \kh{marked} $r$ and an omnian \kh{marked}~$o$ (see Corollary~\ref{cor:wtc_sdn}).}}
\end{figure}

\begin{theorem}\label{thm:stronglyTC}
A multi-semi-directed network~$N$ is strongly tree-child if and only if~$N$ has no omnians.
%~$d^e(v)\geq 2$ or~$d^+(v)=0$ for all~$v\in V$.
\end{theorem}
\begin{proof}
First assume that~$N$ is strongly tree-child. Suppose that there exists a vertex~$v\in V$ that is an omnian, i.e., $d^+(v)\geq 1$ and $d^e(v)\leq 1$. If~$d^e(v)=0$ then all children of~$v$ are reticulations in any rooting of~$N$ and~$v$ is not a leaf since $d^+(v)> 0$. This would contradict that~$N$ is strongly tree-child. Hence, $d^e(v)=1$. By Theorem~\ref{thm:sink_component} (Property (Ci)), $d^-(v)\in\{0,d(v)-1\}$. Since $d^-(v)=d(v)-d^+(v)-d^e(v) \leq d(v) -1 - 1$, this means that $d^-(v)=0$ and thus $d^+(v)=d(v)-1$. 

Let~$e=\{u,v\}$ be the edge \kh{in $N$} incident to~$v$. Observe that all neighbours of~$v$ other than~$u$ are reticulations. Let~$D$ be a tree-child rooting of~$N$. Since~$v$ has at least one child in~$D$ that is not a reticulation, %\kh{and $v$ is an omnian of $N$,
the edge~$e$ is oriented away from~$v$ in~$D$. This means that there are two possibilities. Either~$v$ is a root of~$D$, or \kh{there exists a child $a$ of $v$ in $N$ such that the arc $(v,a)$ of $N$ leaving}
%entering~
$v$ has been subdivided by a root~$w$ \kh{in $D$} and~$D$ contains the arcs~$(w,v)$ \kh{and $(w,a)$.} In the first case, changing the orientation of the arc~$(v,u)$ to~$(u,v)$ (making~$u$ a root) gives a rooting~$D'$ of~$N$ that is not tree-child. \leojuly{Also in the second case,~$D'$ %created \kh{as} in the first case 
is a rooting} of~$N$ that is not tree-child. In both cases, we obtain a contradiction, completing the first direction of the proof.

For the \kh{converse}, \kh{suppose} that \kh{$N$ has no omnians, that is,}~$d^e(v)\geq 2$ or~$d^+(v)=0$ for all~$v\in V$. \kh{Assume} that~$N$ is not strongly tree-child, i.e., there exists a rooting~$D$ of~$N$ and a non-leaf vertex~$v$ of~$D$ such that all children of~$v$ in~$D$ are reticulations. If~$v$ is a reticulation too, then~$d^e_N(v)=0$ and~$d^+_N(v)=1$. If~$v$ is not a reticulation, then~$d^e_N(v)\leq 1$ and~$d^+_N(v)\geq 2$. In both cases, we obtain a contradiction.
\end{proof}

%Call a vertex~$v$ \emph{pink} if $d^+(v)\geq 1$ and $d^e(v)\leq 1$.

%\textbf{Note that this is (almost?) equivalent to \cite{maxfield2025dissimilarity} (but they don't allow roots on edges/arcs), although our result is written in a very different way (which makes it more pleasant to read especially for the strongly tree-child case). See also email conversation.}

%\todo[inline]{@Leo and/or Vincent: Kathi and Niels are a bit worried about whether the degree condition in (ii) and (iii) in the following Proposition is fully correct, since it does not seem to be really used in the proof (at least not explicitly). One direction in the proof seems obvious, but the iff not so much. The reason why this condition was added in the first place is to handle networks like the one in the temporary handdrawn figure, which would otherwise be classified as weakly tree-child, while they are not. In Prop. 11 of \cite{maxfield2025dissimilarity} they also have something similar, which is how Leo found this out. \includegraphics[width=.5\textwidth]{old_figures/weird_example.jpeg}. Leo: I edited the last bit of the proof to explicitly use the degree condition. Niels: Thanks!}

\begin{proposition}\label{prop:weaklyTC}
Let $N=(V,E,A)$ be a $k$-semi-directed network, some $k\geq 1$, with set of reticulations~$R$ and set of omnians~$O$. \nh{Then, $N$ is weakly tree-child if and only if $|O| \leq k$, $d^e(v)\geq 1$ for all~$v\in V$, and there does not exist a \nhh{non-trivial} edge-path between any two vertices in~$O\cup R$.}
%Let $N=(V,E,A)$ be a \kh{$k$-semi-directed network, some $k\geq 1$,} with set of reticulations~$R$ and set of omnians~$O$.
%~$U=\{v\in V\mid d^+(v)\geq 1,\, d^e(v)\leq 1\}$.
%Then the following statements are equivalent.
%\begin{enumerate}
%\item[(i)] $N$ is weakly tree-child;
%\item[(ii)] \kh{$|O|\leq k$, \leomay{$d^e(v)\geq 1$ for all~$v\in V$}, and there does not exist an edge-path between any two vertices in~$O\cup R$.}
%\item[(iii)] $|O|\leq |V|-|R|-|E|$, \leomay{$d^e(v)\geq 1$ for all~$v\in V$}, and there does not exist an edge-path between any two vertices in~$O\cup R$.
%\end{enumerate}
\end{proposition}
\begin{proof}
%(i) $\Rightarrow$ (ii): 
Suppose that~$N$ is weakly tree-child. Then there exists a tree-child rooting~$D$ of~$N$. For each vertex~$u\in O$, either~$u$ is a root in~$D$ or one of the outgoing arcs of~$u$ is subdivided by a vertex that is a root in~$D$. Hence, $|O|\leq k$. \nhh{For the second condition, assume for contradiction that there exists a vertex $v\in V$ with $d_N^e(v) = 0$. Then, \leojuly{by Theorem~\ref{thm:semi-directed-cycle}}, $d_N^+(v) \geq 1$ and in any rooting~$D$ of $N$, every child of $v$ is a reticulation. Thus, $D$ is not tree-child; a contradiction. Lastly, to prove the third condition,}
%\todo{@Leo, please double check the argument.}
\kh{assume for contradiction} that there exists a \nhh{non-trivial} edge-path \kh{$P$} in~$N$ between two vertices in~$O\cup R$, say between~$s$ and~$t$. \leojuly{In a tree-child rooting, an edge incident to a reticulation must be oriented away from the reticulation and, similarly, an edge incident to an omnian must be oriented away from the omnian. Hence, the edge of~$P$ incident to~$s$ is oriented away from~$s$ in~$D$. Similarly, the edge of~$P$ incident to~$t$ is oriented away from~$t$ in~$D$.} 
%Then, in~$D$, the path \kh{$P$} must \leojuly{contain} an outgoing arc of~$s$ and %end 
%\todo{\kh{I am not sure "end"is the correct word here. It might be better to say "can be extended by". This would not change the remainder of the proof.}}
%with
%an outgoing arc of~$t$ \leojuly{(since, in a tree-child rooting, an edge incident to a reticulation must be oriented away from the reticulation and, similarly, an edge incident to an omnian must be oriented away from the omnian)}.
\kh{Therefore, $P$ must contain a reticulation in $D$ and hence in~$N$;} a contradiction since no internal vertex of an edge-path can be a reticulation.

%\kh{(ii) $\Rightarrow$ (iii): It suffices to show that 
%$k= |V|-|R|-|E|$. But this %holds by Corollary~\ref{cor:roots}. }

%(iii) $\Rightarrow$ (i): 
\nh{For the other implication, suppose that $|O|\leq k$}, \kh{that} \leomay{$d^e(v)\geq 1$ for all~$v\in V$}, and that there does not exist a \nhh{non-trivial} edge-path between any two vertices in~$O\cup R$. By Lemma~\ref{lem:rootings}, there exists a rooting~$D$ of~$N$ such that each vertex in~$O$ is a root in~$D$. Suppose that~$D$ is not tree-child. Let~$v$ be a non-leaf vertex of~$D$ all whose children in~$D$ are reticulations. Then, we obtain a contradiction since then either $d^e_N(v)=0$ or~$v$ is an omnian of~$N$ that is not a root of~$D$.
\end{proof}

%\todo[inline]{Niels: Idea, we could consider adding the following corollary of this proposition where we specifically characterize weakly-tree-child for 1-semi-directed networks, to set us apart a bit more from \cite{maxfield2025dissimilarity}. In particular, the characterization is a lot nicer then and most people are concerned with sd-networks only any way.  Corollary: ``A semi-directed network $N$ is weakly tree-child if and only if it has at most one omnian $o$ and there is no edge-path between $o$ and a reticulation $r$.''}

\leojuly{We can simplify the characterization in \kh{Proposition~\ref{prop:weaklyTC}} for semi-directed networks as follows, using that, in a (multi-)semi-directed network, $d^e (v) \geq 1$ \kh{holds} for any non-omnian vertex $v$ and there is no \kh{non-trivial} edge-path between two reticulations (Theorem~\ref{thm:semi-directed-cycle}).}
\begin{corollary}\label{cor:wtc_sdn}
    \nhh{A semi-directed network~$N$ is weakly tree-child if and only if $N$ has at most one omnian and if there is an omnian~$o$, then $d^e (o) \geq 1$ and there is no \nhh{non-trivial} edge-path \leojuly{between} $o$ \leojuly{and} a reticulation~$r$ of $N$.}
\end{corollary}

%\subsection{Tree-Based}
\kh{We next turn our attention to tree-based multi-semi-directed networks.}
A \emph{spanning tree} of a mixed graph~$G$ is a subtree of~$G$ that contains all vertices of~$G$ and is a tree. 
A multi-rooted network~$N$ is \emph{tree-based} if it has a rooted spanning tree that has the same leaf set as~$N$\remove{(see~[??])}. \add{Clearly, a multi-rooted network that is tree-based must be a rooted network, since the corresponding spanning tree must have a single root.} \add{Note that binary rooted tree-based networks were introduced in~\cite{francis2015which} and extended to non-binary rooted networks in~\cite{jetten2016nonbinary}, who also defined the stricter notion of strictly tree-basedness, which we do not consider here.} A multi-semi-directed network~$N$ is \emph{weakly tree-based} if~$N$ has a rooting that is tree-based and~$N$ is \emph{strongly tree-based} if every rooting of~$N$ is tree-based, see Figure~\ref{fig:tree-based} for an example.\remove{Clearly} \add{Since every multi-rooted tree-based network is a rooted network}, all multi-semi-directed networks that are weakly \add{or strongly} tree-based are semi-directed networks. Hence, we will focus on semi-directed networks \kh{for the remainder of} this section.

\begin{figure}[htb]
    \centering
    \includegraphics[]{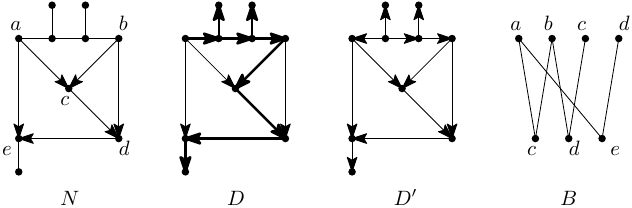}
    \caption{\label{fig:tree-based} A semi-directed network~$N$ that is weakly tree-based but not strongly tree-based, along with a rooting~$D$ \kh{for it} that is tree-based (a spanning tree with leaf set $L(D)$ is indicated in bold) and a rooting~$D'$ \kh{for it} that is not tree-based. \emph{Far right:} The bipartite graph~$B$ used in the proof of Theorem~\ref{thm:strongTB}.}
\end{figure}

%A \emph{sink} of a subtree is a vertex that has two???

%\begin{observation}
%    A subtree of a multi-semi-directed network has a rooting if and only if it does not have a sink.
%\end{observation}

\kh{To state our next result, we define for a multi-semi-directed network $N=(V,E,A)$ and a subset $S$ of the 
set of omnians of $N$ the \leojuly{set}} % quantity
$$
\delta^+(S):=\{t\in V\mid (s,t)\in A \mbox{ for some } s\in S\}.
$$ 
 \leojuly{Hence}, $\delta^+(S)$ \leojuly{contains} the %number of
 vertices of $N$ that are the head of an arc that starts at an omnian %that is
 contained in~$S$. 
 \nh{Note that by definition, $\delta^+ (S)$ may contain vertices that are also in~$S$.} 
 %Furthermore, we say for a subset $S\subset V$ that a vertex $w\in V\setminus S$ is a {\em neighbor of $S$ in $ V\setminus S$}  if there exists a vertex in $S$ for which $w$ is a neighbor.}
 % Leo: I don't think it is necessary to include the previous sentence.

\begin{theorem}\label{thm:strongTB}
Let~$N$ be a %multi-rooted
semi-directed network with $O$ its set of omnians. Then~$N$ is strongly tree-based if and only if for each $S\subseteq O$ we have that $|\delta^+(S)|\geq |S|$.
\end{theorem}
\begin{proof}
    \nh{Let $R$ be the set of reticulations of $N$.} We follow a similar approach as in~\cite{jetten2016nonbinary} for directed networks. \kh{More precisely, we first associate}
a bipartite graph~$B=(V',E')$ \kh{to $N$ that has} vertex set~$V'=R'\cup O'$, with~$R'$ containing a copy of each reticulation in~$R$ and $O'$ containing a copy of each omnian in~$O$. Hence, a vertex \kh{of $N$} that is a reticulation as well as an omnian has two corresponding vertices in~$V'$. \kh{For $r\in R'$ and~$o\in O'$, we define} 
%edge set~$E'$ contains an edge
$\{r,o\}$ \kh{to be an edge in $E'$} if
%the corresponding vertices are connected by an arc 
$(o,r)$ \kh{is an arc} in~$N$.
\kh{Note that,} by Hall's marriage theorem~\cite{hall},~$B$ has a matching \kh{that covers}~$O$ if and only if each \kh{subset} $S\subseteq O'$ has at least~$|S|$ neighbors in~$B$.

To prove the theorem, assume first that~$N$ is strongly tree-based. Consider any rooting~$D$ of~$N$ with the same vertex set as~$N$ (which exists because any rooting where the root subdivides an edge or arc of~$N$ can be easily modified to \kh{obtain} a rooting in which the root is a vertex of~$N$). \kh{Let $\rho$ denote the sole root of $D$}. We create a rooting~$D'$ \kh{of $N$} from~$D$ as follows. If~$\rho$ is an omnian in~$N$ and $d_N^e(\rho)=1$ 
then subdivide the edge incident to~$\rho$ in~$N$ by a new vertex~$\rho'$, make~$\rho'$ the root of~$D'$, orient the edge between~$\rho$ and~$\rho'$ as~$(\rho',\rho)$, \kh{and retain the directions of the remaining arcs of $D$}. Otherwise, \kh{define}~$D'$ \kh{to be} $D$. By construction of~$D'$, for each omnian~$o\in O$ of~$N$, \kh{we have} that all outgoing arcs of~$o$ in~$D$ are also outgoing arcs of~$o$ in~$N$.
    Since~$N$ is strongly tree-based,~$D'$ is tree-based. Consider a base tree~$T'$ of~$D'$. Then,~$T'$ contains at least one outgoing arc $a_o$ (in~$D'$), \kh{for each~$o\in O$ (the omnians of~$N$)}. In addition,~$T'$ contains exactly one incoming arc of each reticulation in $R$. Hence, the \kh{arcs $a_o$, $o\in O$,} form a matching in~$B$ that covers~$O$. By Hall's marriage theorem \kh{recalled above}, this implies that each $S\subseteq O$ has at least~$|S|$ neighbors in~$B$ and hence \kh{that} $|\delta^+(S)|\geq |S|$.

    \kh{Conversely}, assume that for each $S\subseteq O$ we have that $|\delta^+(S)|\geq |S|$. \kh{Then}, by Hall's marriage theorem,~$B$ has a matching~$M$ \kh{that covers}~$O$. Consider any rooting~$D$ of~$N$. \kh{Let $\rho$ denote the sole root of $D$.} We construct a rooted spanning tree~$T$ of~$D$ with leaf set~$L(D)$ as follows, see Figure~\ref{fig:tree-based}. For each~$o\in O$ consider the reticulation~$r\in R$ that~$o$ is matched to by~$M$. Note that~$(o,r)$ may not be an arc of~$D$ because it could have been subdivided by the root. If~$(o,r)$ is an arc of~$D$, then include it in~$T$. Otherwise, $D$ contains an arc $(\rho,r)$ which we include in~$T$. For each \kh{reticulation in} $ R$ that does not have an incoming arc in~$T$ yet, choose one incoming arc arbitrarily and add it to~$T$. Finally, add all arcs \kh{whose heads are not reticulations}
    %that are not reticulation arcs 
    also to~$T$.  Clearly,~$T$ is a rooted spanning tree of~$D$. 
    
    It remains to show that~$T$ has leaf set~$L(D)$. Clearly, each leaf of~$D$ is a leaf of~$T$. Suppose~$T$ has a leaf~$v$ \kh{that is not a leaf in} $D$. Then, \kh{in $D$,} $v$ has at least one outgoing arc and \kh{the head of every outgoing arc of $v$ is a reticulation.}
    %all outgoing arcs of~$v$ are reticulation arcs. 
    Hence, in~$N$, $v$ has exactly one incident edge, no incoming arcs and at least one outgoing arc. \kh{Thus,}~$v$ is an omnian of~$N$. This leads to a contradiction since~$T$ contains, for each omnian of~$N$, at least one outgoing arc of~$D$.
\end{proof}

\kh{Intriguingly, characterizing semi-directed networks that are weakly tree-based requires a different approach}
\vm{(see Corollary~\ref{cor:weak-tree-based} \leojuly{at the end of Section~\ref{sec:paths}}}).

\section{Cherry Picking} \label{sec:cherries}

\vm{In this section we consider another concept, called ``cherry picking'', that can be used to characterize certain multi-semi-directed networks.
More specifically, we focus on how cherry picking can be applied to characterize weakly and strongly orchard multi-semi-directed networks,
classes of networks that can be used to model lateral gene transfer (see e.g. \cite{van2022orchard}).}
%\kh{Motivated by the fact that the important evolutionary process of horizontal gene transfer is not only suspected to take place between bacteria that live within the same ecological niche (and so can be modeled in terms of a rooted phylogenetic network) but also between bacteria living in distinct ecological niches (and so can be modeled in terms of multi-rooted networks), we next turn our attention to the problem of what can be said about classes of multi-semi-directed networks that might turn out to be useful in this context, that is, the classes of weakly and strongly orchard multi-semi-directed networks.   }
\kh{To state the main results  (Theorems~\ref{thm:orchard} and \ref{thm:strongorchard}),  we require some  further definitions.} \leojuly{Recall that a \emph{cherry} of a mixed graph is an ordered pair of leaves~$(x,y)$ such that there is a length-$2$ path between~$x$ and~$y$, either consisting of two edges or of two arcs directed towards~$x$ and~$y$ respectively. \emph{Reducing} a cherry $(x,y)$ is defined as deleting~$x$ and suppressing 
any resulting non-root degree-2 vertex.}
%\vm{Note that in this section, to simplify the exposition, we shall restrict our attention to multi-rooted and  multi-semi-directed networks that do not have isolated vertices as connected components.} 

%\textbf{In the discussion of \cite{maxfield2025dissimilarity} they argue that reticulated cherries and cherries in sd-networks are well-defined and that hence, it would be interesting to look at cherry picking sequences for sd-networks. They don't prove any results, but we could mention this.}

%A \emph{cherry} of a mixed graph $G$ is an ordered pair of leaves~$(x,y)$ such that there is \leojuly{a length-$2$ path} %edge-path joining $x$ and $y$.
\leojuly{To introduce cherry picking for multi-rooted networks, consider a multi-rooted network~$D$ on~$X$.}
%A cherry of~$D$ is an ordered pair of leaves~$(x,y)$ such that there is a length-$2$ path %edge-path
%joining $x$ and~$y$.
% this is not an edge-path in the rooted case
%assume that $G$ is a multi-rooted network on $X$.}
A \emph{reticulated cherry}\footnote{\add{Note that this definition is only for multi-rooted networks. Reticulated cherries in multi-semi-directed networks are defined in a slightly different way below.}} of~$D$
%or of a $1$-semi-directed network
is an ordered pair of leaves~$(x,y)$ such that~$y$ is a reticulation leaf and there is a length-$3$ path between~$x$ and~$y$. \emph{Reducing} a reticulated cherry \kh{ $(x,y)$ of $D$} means deleting the arc \kh{from the parent of $x$ to the parent of $y$} and suppressing any resulting non-root degree-2 vertices. Let~$D(x,y)$ be the %multiple rooted directed graph that is the
result of reducing a cherry or reticulate cherry~$(x,y)$ in~$D$. 
\kh{Then the leaf set of $D(x,y)$ is either the same as the leaf set of~$D$\add{,} or, if $(x,y)$ is a cherry of~$D$, %and reducing $(x,y)$ led to deleting~$x$,
then the leaf set of $D(x,y)$ is $X\setminus\{x\}$.} %Note that each connected component of $D(x,y)$ is} %an arc whose head is a vertex in $X$, or 
%\leojuly{a multi-rooted network on some subset of $X$, possibly consisting of a single arc whose head is in~$X$.}
%multi-rooted network that is the result of reducing a cherry or reticulate cherry~$(x,y)$ in~$\kh{G}$. 
%\kh{Then the leaf set of $G(x,y)$ is either the same as the leaf set of $G$ or, if $(x,y)$ is a cherry of $G$ and reducing $(x,y)$ led to deleting $x$ then the leaf set of $G(x,y)$ is $X-\{x\}$. }
We say that
$D$
%%A multi-rooted network~$N$ \kh{on $X$} 
is \emph{orchard} if it can be reduced to a %\kh{graph that is the disjoint union of arcs whose respective heads are in $X$ using} 
disjoint union of arcs using
a sequence of cherry reductions and reticulated cherry reductions.

\kh{Furthermore, we say that} a multi-semi-directed network~$N$ 
is \emph{weakly orchard} if~$N$ has a rooting that is orchard and \kh{that}~$N$ is \emph{strongly orchard} if every rooting of~$N$ is orchard. See Figure~\ref{fig:weakstrongorchard} for an example of a semi-directed network~$N$ %\kh{on $X=\{a,\ldots, f\}$}
that is weakly orchard but not strongly orchard.

\begin{figure}%[htb]
    \centering
    \includegraphics{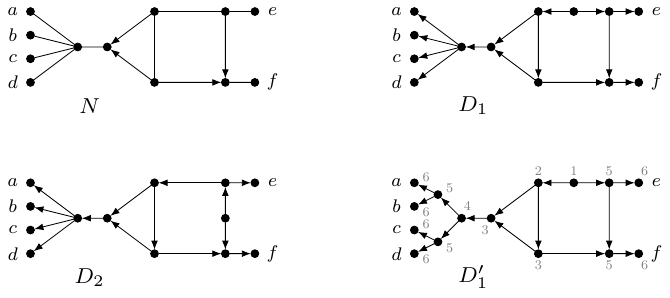}
    \caption{\nh{A semi-directed network~$N$ on $X=\{a,b,c,d,e,f\}$ that is weakly orchard but not strongly orchard, along with a rooting~$D_1$ that is orchard and a  rooting~$D_2$ that is not orchard. %In each case, the root is indicated by a vertex of indegree one and outdegree one. 
    \kh{The rooted network} $D'_1$ is a binary resolution of $D_1$ with the gray numbers indicating an HGT-consistent labelling for $D'_1$.\label{fig:weakstrongorchard}}}
\end{figure}

We now define reticulated cherries and their reductions for a multi-semi-directed network \kh{$N$} on~$X$. This definition is slightly different from the definition above for \kh{multi-rooted} networks since~$N$ may have elements of~$X$ \leojuly{that are roots,}
%\kh{that, in any rooting $D$ of $N$, must either become a root of~$D$ or must become adjacent with a root of $D$,
see for example %leaf
\leojuly{root}~$e$ in the $2$-semi-directed network in Figure~\ref{fig:defs}. 
%Suppose for the following that $N$ is a multi-semi-directed network~$N$ on~$X$. 
%\leojuly{A \emph{cherry} of~$N$ is an ordered pair~$(x,y)$ with~$x,y\in X$ such that there is a length-$2$ edge-path between~$x$ and~$y$.}
% in the semi-directed case it has to be an edge path (not an arc followed by an edge, that would be a reticulated cherry)
A \emph{reticulated cherry} of $N$
%a multi-semi-directed network~$N$ on~$X$ 
is an ordered pair~$(x,y)$ with~$x,y\in X$ such that~$y$ is a reticulation leaf and there is a length-$2$ or length-$3$ path between~$x$ and~$y$. 
%\todo{It is common that~$x$ is the reticulation leaf} 
\emph{Reducing} a reticulated cherry \kh{$(x,y)$ of $N$ is defined as} deleting the arc between the neighbours of~$x$ and~$y$ %\kh{in case it is a length-3 path and}
or deleting the arc between~$x$ and the neighbour of~$y$ %\kh{otherwise, each time}
and suppressing any resulting non-root degree-2 vertices. See Figures~\ref{fig:orchard1} and~\ref{fig:orchard2} for examples. 
\kh{For  $(x,y)$ a reticulated cherry or a cherry of $N$,} let~$N(x,y)$ denote the resulting mixed graph that is the result of reducing~$(x,y)$ in~$N$.
%Let~$N(x,y)$ denote the multi-semi-directed network that is the result of reducing~$(x,y)$ in~$N$. 
%\kh{Note that the leaf set of $N(x,y)$ is the leaf set of $N$ if $(x,y)$ is a reticulate cherry of $N$ and that the leaf set of $N(x,y)$ is $X-\{x\}$ if $(x,y)$ is a cherry of $N$ and $x$ was deleted when reducing $(x,y)$.
% Leo: the above is not necessarily true. Leaves may become roots.
Note that each connected component of $N(x,y)$ is either an isolated vertex in $X$ %, an edge between two vertices from $X$,
or a multi-semi-directed network on a subset of $X$.

\kh{Intriguingly, binary rooted networks that are orchard can be characterized in terms of a so-called HGT-consistent labelling \cite{van2022orchard}. As it turns out, this concept can be canonically extended to binary multi-rooted networks. In turn, this provides us with a tool to characterize  multi-semi-directed networks that are weakly orchard. To make this more precise, suppose that \nh{$D=(V,A,\kh{\emptyset})$} is a binary multi-rooted network.
Then we call a map $t:V\rightarrow\mathbb{N}$ a} \emph{HGT-consistent labelling} \kh{for $D$} if
\begin{enumerate}
    \item $t(u)\leq t(v)$ for each arc~$(u,v)\in A$;
    \item $t(u) < t(v)$ for each arc~$(u,v)\in A$ with~$v$ not a reticulation; and
    \item for each reticulation~$v$ it holds that $t(u) = t(v)$ for exactly one parent~$u$ of~$v$.
\end{enumerate}
\nh{See Figure~\ref{fig:weakstrongorchard} for an example of a HGT-consistent labelling of a 1-rooted network.}

\kh{We call a} rooting~$D$ of a multi-semi-directed network~$N$ \kh{on $X$} \emph{nice} %\kh{if, for every connected component $C$ of $N$ that is not an isolated vertex, every root of $C$ in the rooting $D$}
if every root of~$D$ is either a vertex of~$N$ that is not in~$X$ or subdivides an arc of~$N$ whose tail is in~$X$.
\vm{We will use the following useful fact concerning nice rootings.}  %% if every conn component that is not an isolated vertex has a nice rooting.

%%%%%%%%%%%%%%%%%%%%%%%%%%%%%%%%%%%%%%%%%%%
%%%%%%%%% begin: old version of the definition of "nice" %%%%%

%\kh{Calling a} rooting~$D$ of a multi-semi-directed network~$N$ \kh{on $X$} \emph{nice} if every root of~$D$ is either a vertex of~$N$ that is not in~$X$ or subdivides an arc of~$N$ whose tail is in~$X$, \kh{we obtain the following result.}  %% if every conn component that is not an isolated vertex has a nice rooting.

%%%%%%%%% end: old version of the definition of nice %%%%%
%%%%%%%%%%%%%%%%%%%%%%%%%%%%%%%%%%%%%%%%%%%%%%%%%

\begin{observation}
 \label{obs:nicerooting} \kh{Every} multi-semi-directed network has a nice rooting.
\end{observation}
\begin{proof}
\kh{Suppose that $N$ is a multi-semi-directed network.}
Consider an arbitrary rooting~$D$ \kh{of $N$}. If a root~$\rho$ of~$D$ is not a vertex of~$N$, then it subdivides an edge~\add{$e$} or an arc~\add{$a$} of~$N$.
%\kh{since, by our general assumption for this section, $N$ cannot contain a connected component in the form of an isolated vertex}.
If~$\rho$ subdivides~\kh{$e$}, we can make any \kh{vertex incident} with \kh{$e$} a root instead of~$\rho$. If~$\rho$ subdivides an arc~$a$ whose tail is not in~$X$,
%\kh{the tail of $a$ is not in $X$}
then we can make the tail \kh{of $a$} a root instead of~$\rho$.
\end{proof}

\vm{We now define a {\em binary resolution} of a multi-rooted network $N=\kh{(V,A)}$ which we 
shall need to state our next result.
This is the binary multi-rooted network obtained from $N$ by 
(i) replacing every vertex~$v$ with $d_N^+(v)\geq 3$ and its set of outgoing arcs 
by a rooted binary tree $T$ with root $v$,
so that all arcs in $T$ are directed away from $v$, and the leaf set 
of $T$ consists of those $w\in V$ such that $(v,w)\in A$,
and (ii) replacing every vertex $w$ in $N$ with  $d_N^-(w)\geq 3$ and its set of incoming arcs 
by a rooted binary tree $T$ with root $w$, in
which the directions of all arcs in $T$ are reversed so that they are all directed towards $w$ and 
the leaf set of $T$ consists of those $v\in V$ with $(v,w)\in A$.}
\nh{See Figure~\ref{fig:weakstrongorchard} for an example.}

\leojuly{To simplify the exposition of the remainder of this section, we shall from now on %restrict our attention to multi-rooted and  multi-semi-directed networks that do not
assume, without loss of generality, that multi-rooted and multi-semi-directed networks have no isolated vertices.}
%\todo{Short question, did we allow isolated leaves in networks before this section? If so, then this is all correct. (Update: I just had a look, and it seems we do, so I guess this is indeed correct.)}
% Leo: yes in principle isolated vertices are allowed

\begin{theorem}\label{thm:orchard}
 Given a multi-semi-directed network~$N$ on \kh{$X$}, the following are equivalent.
 \begin{enumerate}
     \item[(1)] $N$ is weakly orchard;
     %\item[(2)] $N$ is strongly orchard;
     \item[(2)] there exists a sequence of cherry reductions and reticulated cherry reductions that reduces~$N$ to 
     \kh{a forest in which each tree is either a single edge whose two adjacent vertices  are in~$X$ or a single vertex that is in~$X$;}
     %a graph in which each vertex~$v$ has~$d(v)=0$ or~$d(v)=d^e(v)=1$;
     \item[(3)] $N$ has a rooting that has a binary resolution that admits an HGT-consistent labelling;
     \item[(4)] every nice rooting of~$N$ is orchard;
     \item[(5)] $N(x,y)$ is weakly orchard for some cherry or reticulated cherry~$(x,y)$.
 \end{enumerate}
\end{theorem}
\begin{proof}
%We start by showing that (1) implies (4). If~$N$ is weakly orchard then it has a rooting~$D$ that is orchard. Let~$D'$ be a nice rooting. It is straightforward to verify that any sequence of cherry and reticulated cherry reductions that reduces~$D$ to a disjoint union of arcs also reduces~$D'$ to a disjoint union of arcs.

That (4) implies (1) is trivial, given that~$N$ has at least one nice rooting by Observation~\ref{obs:nicerooting}.

We now show that (1) implies (3). If~$N$ is weakly orchard then it has a rooting~$D$ that is 
%a multi-rooted directed network and is 
orchard. \kh{Since $D$ is a multi-rooted network, we can obtain a rooted network} $D'$ from~$D$ by adding a new root~\kh{$\rho$} with an arc to each root or~$D$. Consider any sequence of cherry and reticulated cherry reductions that transforms~$D$ into a disjoint union \kh{$U$} of arcs \kh{the heads of which are elements in~$X$.} The same sequence transforms~$D'$ into a rooted star tree \kh{$S$, that is, the tree obtained from $U$ as follows. First, add a root $\rho'$ to $U$. Next, add an arc from $\rho'$ to the tail $t_a$ of each arc $a$ in $U$. Finally, suppress all vertices $t_a$. Clearly, $S$}
can easily be transformed into a single arc whose head is in $X$ by cherry reductions. Hence,~$D'$ is also orchard and, by \cite[Theorem~2]{van2022orchard}, $D'$ has a binary resolution~$D'_b$ that admits a HGT-consistent labelling. Deleting all \kh{vertices and arcs of $D'_b$ that where added to $D'$ to obtain a binary resolution of $\rho$ as part of obtaining  $D_b'$} 
%correspond to the binary resolution of~$\rho$ 
then gives a binary resolution of~$D$ with a HGT-consistent labelling, completing the proof that (1) implies (3). The \kh{converse} direction, i.e. that (3) implies (1), can be shown \kh{in a similar manner}.

We now prove that (1) implies (2). Let~$D$ be a rooting of~$N$ \kh{that is orchard}. Consider a sequence \kh{$\sigma$} of cherry and reticulated cherry reductions that transforms~$D$ into a disjoint union of arcs \kh{the heads of which are elements in $X$}.  A cherry $(x,y)$ that can be reduced in~$D$ can also be reduced in~$N$ unless the path between~$x$ and~$y$ in~$D$ contains a root of~$D$ and this root is not a vertex of~$N$ (i.e. it is suppressed \kh{when semi-deorienting $D$ to obtain $N$}). In this case,~$N$ contains a connected component consisting of only the edge~$\{x,y\}$. A reticulated cherry $(x,y)$ that can be reduced in~$D$ can also be reduced in~$N$. The \kh{outcome of reducing $(x,y)$ in $N$ and in $D$ in this case, respectively,} is the same (up to taking the semi-deorientation) unless the path between~$x$ and~$y$ in~$D$ contains a root of~$D$ and this root is not a vertex of~$N$ (i.e. it is suppressed \kh{when semi-deorienting $D$ to obtain $N$}). In this case, the reduced version of~$D$ contains a connected component consisting of a single arc \kh{whose head is in $X$}, while the reduced version of~$N$ contains a connected component consisting of a single vertex \kh{in $X$}. Hence, 
%the same sequence of cherry and reticulated cherry reductions 
\kh{$\sigma$ reduces}
$N$ to a forest in which each connected component is either  %with two vertices and a single edge or a single vertex.
\kh{a single edge whose two incident vertices  are in $X$ or a single vertex that is in $X$}.

We now show that (2) implies (4). Let~$D$ be a nice rooting of~$N$. We prove by induction on the number \kh{$m$} of vertices of~$D$ that~$D$ is orchard. \kh{Since $|X|\geq 2$, the bases case is $m=2$ and} is trivial. \kh{Assume that the stated implication holds for all multi-semi-directed networks that have a nice rooting with $2\leq l\leq m$  vertices and that $N$ is such that $D$ has $m+1$ vertices.} Let~$(x,y)$ be a cherry or reticulated cherry in~$N$. Then~$(x,y)$ is also a cherry or reticulated cherry of~$D$ \kh{because $D$ is a nice rooting of $N$}. Note that it is possible that~$N$ contains a length-$2$ path between~$x$ and~$y$ while~$D$ contains a length-$3$ path between~$x$ and~$y$. This can happen if \leojuly{$d(x)=d^+(x)=1$.} %\kh{for one of $x$ and $y$, say $x$, we have that  $d^e_N(x)=0$.} %\todo{Please remove comment if you agree. Otherwise we need to discuss later today
%should this be $d_N (x)$ ? and how can $d_N (x) = 0$ when we just say that $x$ is part of a length-2 path in $N$?} %is a root of $N$.  
In this case, \kh{$d_{N(x,y)}(x)=0$ and so $x$ is an isolated vertex in~$N(x,y)$} while~$D(x,y)$ contains a connected component 
\kh{in the form of an arc from a root of $D(x,y)$ to $x$.} 
%with two vertices, a root and~$x$, connected by an arc. 
%\kh{Since in the definition of a nice rooting of a multi-semi-directed network a connected component in the form of an isolated vertex 
%does not} \vm{need to be taken into account,} \kh{it follows that}
\leojuly{In this case, let~$D',N'$ be $D(x,y),N(x,y)$, respectively, with the connected component containing~$x$ removed. Otherwise, simply let~$D'=D(x,y)$ and~$N'=N(x,y)$. Then~$D'$ is a nice rooting of~$N'$.}
%\todo{What is the difference between the first and second sentence? In both cases we set $D' = D(x,y)$ and $N' = N(x,y)$. Update: I think I understand now, in both cases, we set $D' = ..$ but in the first case we additionally deleted $x$. The sentences just looked the same to me at first glance ;), so perhaps add ``Otherwise, simply let ...''}
%$D(x,y)$ is a nice rooting of~$N(x,y)$.
\kh{By} induction, \kh{it follows that}~\leojuly{$D'$ is orchard, from which we can conclude that}~$D$ is orchard. 

By the equivalence of (1) and (2) shown \kh{above}, it follows easily that (5) implies (1).

It remains to prove that (1) implies (5). Suppose \kh{that}~$N$ is weakly orchard and \kh{that}~$(x,y)$ is a cherry or reticulated cherry of~$N$. By Observation~\ref{obs:nicerooting}, $N$ has a nice rooting~$D$. Then~$D$ is orchard by the \kh{equivalence of} (1) and (4). Furthermore,~$(x,y)$ is a cherry or reticulated cherry also in~$D$.
Let~$D'$ be the rooted network obtained from~$D$ by adding a new root~$\rho$ with an arc to each previous root and suppressing any resulting non-root degree-$2$ vertices. Then~$(x,y)$ is also a cherry or reticulated cherry in $D'$ and~$D'$ is also orchard. Hence, by~\cite[Proposition~1]{janssen2021cherry}, $D'(x,y)$ is orchard. Thus,~$D(x,y)$ is orchard. It follows that~$N(x,y)$ has a rooting \kh{that is orchard}. Hence,~$N(x,y)$ is weakly orchard.
\end{proof}

\begin{figure}[htb]
    \centering
    \includegraphics{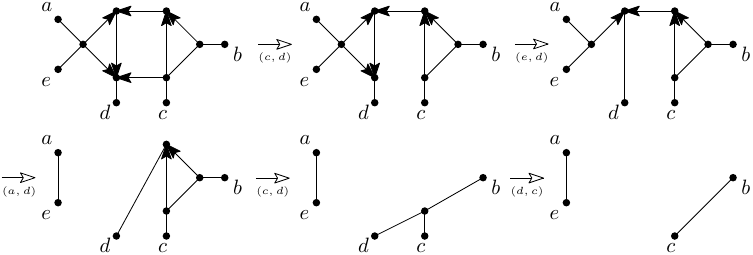}
    \caption{Example of a 
    %multi-semi-directed weakly orchard 
    \kh{weakly orchard multi-semi-directed}
    network \kh{on $\add{X=}\{a\ldots, e\}$} and a sequence of cherry reductions and reticulated cherry reductions. \nh{In each case, the (reticulated) cherry that is reduced is indicated below the arrow that indicates the reduction.}}% \fix{Kathi says: It might be nice to include the reticulate cherries/cherry that is reduced in the fig?}}
    \label{fig:orchard1}
\end{figure}

\begin{figure}[htb]
    \centering
    \includegraphics{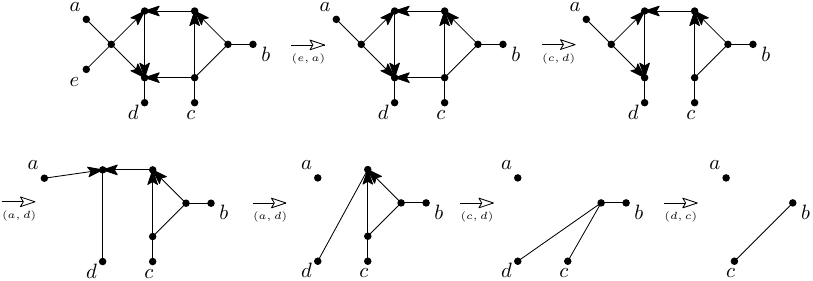}
    \caption{An alternative sequence of cherry reductions and reticulated cherry reductions for the weakly orchard multi-semi-directed
    network on \add{$X=\{a\ldots, e\}$} from Figure~\ref{fig:orchard1}. In each case, the (reticulated) cherry that is reduced is indicated below the arrow that indicates the reduction. \add{Note that in the bottom-left multi-semi-directed network, $a$ has become a root instead of a leaf and that $(a,d)$ is a reticulated cherry in that network since $a \in X$.}}
    %\fix{Kathi says: It might again be good to include the reticulate cherry/cherry to which the reduction is applied in the figure(?) In this case we should update the caption accordingly.}}
    \label{fig:orchard2}
\end{figure}

From  \kh{Theorem~\ref{thm:orchard}}, it follows that one can decide in \leo{linear} time whether a given multi-semi-directed network is weakly orchard since, just as for rooted networks, cherries and reticulated cherries can be reduced in arbitrary order.

%\clearpage

%\begin{proposition}
%Let~$N$ be a multi-semi-directed network on~$X$. If~$N$ is weakly orchard and~$(x,y)$ is a cherry or reticulated cherry of~$N$, then~$N(x,y)$ is weakly orchard.
%\end{proposition}
%\begin{proof}
%Suppose~$N$ is weakly orchard and~$(x,y)$ is a cherry or reticulated cherry of~$N$. 

%By Observation~\ref{obs:nicerooting}, $N$ has a nice rooting~$D$. Then~$D$ is orchard by Proposition~\ref{prop:orchard} and~$(x,y)$ is a cherry or reticulated cherry also in~$D$.
%unless~$(x,y)$ is a reticulated cherry with a length-$2$ path between~$x$ and~$y$ in~$N$. In the latter case, this path traverses an arc leaving~$x$ and an edge incident to~$y$. In this case, modify the rooting~$D$ by subdividing the arc leaving~$x$ by a root and orienting the edge incident to~$x$ towards~$x$. This modified rooting contains reticulated cherry~$(x,y)$. Hence, in any case, $(x,y)$ is a cherry or reticulated cherry in~$D$ and~$D$ is orchard. 
%Let~$D'$ be the rooted network obtained from~$D$ by adding a new root~$\rho$ with an arc to each previous root and suppressing any resulting non-root degree-$2$ vertices. Then~$D'$ also has a cherry or reticulated cherry $(x,y)$ and~$D'$ is also orchard. Hence, by~\cite[Proposition~1]{janssen2021cherry}, $D'(x,y)$ is orchard. Hence,~$D(x,y)$ is orchard. It follows that~$N(x,y)$ has an orchard rooting. Hence,~$N(x,y)$ is weakly orchard.
%\end{proof}

%For the remainder of this section, we focus on semi-directed networks.
\vm{To state the second main result of this section 
which concerns semi-directed networks, we require some further definitions.} \kh{Suppose $N$ is a semi-directed network on $X$. A}
\emph{cherry picking sequence} of~$N$ is a sequence $(\kh{s_1},\ldots ,s_k)$ of ordered pairs of elements of~$X$, such that~$\kh{s_1}$ is a cherry or reticulated cherry of~$N_0:=N$ and, for all~$i\in\{1,\ldots ,\kh{k-1}\}$, the pair~$s_{i+1}$ is a cherry or reticulated cherry in \kh{$N_i:=N_{i-1}(s_i)$}
%~$N\circ (s_1,\ldots ,s_{i-1})$ and~$N\circ (s_1,\ldots ,s_k)$ 
\kh{and $N_k:=N_{k-1}(s_k)$ is a graph in which each connected component is either an isolated vertex in $X$ or an edge such that both incident vertices are contained in $X$.} 
% or two vertices with one edge.
\kh{Note that in case we want to emphasize the order in which the pairs $s_i$ are reduced, we also write \leojuly{$N\circ(s_1,\ldots, s_i)$ for $N_i$.} % in case  $1\leq i\leq k-1$, 
%and $N\circ(s_0,\ldots, s_k)$ for $N_k(s_k)$. 
}

%\todo{``with $N_0 = N$.''?}
% Leo: we already say this above

%Observe that a cherry-picking sequence of a semi-directed network~$N$ necessarily reduces~$N$ to a graph with two vertices and one edge.

\begin{lemma}\label{lem:remainsstronglyorchard}
Let~$N$ be a semi-directed network and \kh{let} $(x,y)$ \kh{be} a cherry or \kh{a} reticulated cherry of~$N$. If~$N$ is strongly orchard then~$N(x,y)$ is strongly orchard.
\end{lemma}
\begin{proof}
    Let~$D$ be any rooting of~$N(x,y)$. Then a rooting~$D'$ of~$N$ can be obtained from~$D$ by \kh{the following} small modifications. If~$(x,y)$ is a cherry \kh{and $x$ was deleted in the construction of $N(x,y)$ from $N$}, then~$D'$ is obtained from~$D$ by subdividing the arc incident to~$y$ by a vertex~$w$ and adding leaf~$x$ with an arc $(w,x)$. If~$(x,y)$ is a reticulated cherry, then~$D'$ is obtained from~$D$ by subdividing the arc incident to~$y$ by a vertex~$v$ and the arc incident to~$x$ by a vertex~$u$ and adding an arc $(u,v)$. Since~$N$ is strongly orchard,~$D'$ is orchard. It then follows that~$D$ is orchard. Hence,~$N(x,y)$ is strongly orchard.
\end{proof}

The converse of the \kh{Lemma~\ref{lem:remainsstronglyorchard}} does not hold in general, see Figure~\ref{fig:weak_strong_orchard_sequence}. Motivated by this, we define an \emph{\scr-cherry} (source component reticulated cherry) of a semi-directed network~$N$ as a reticulated cherry $(x,y)$ such that %the neighbour of
$x$ is in a source component of the network.

The next lemma, combined with \kh{Lemma~\ref{lem:remainsstronglyorchard}}, shows that cherries and reticulated cherries that are not \scr-cherries can be reduced in arbitrary order. 

\begin{lemma}\label{lem:remainsnotstronglyorchard}
Let~$N$ be a semi-directed network  and \kh{let} $(x,y)$ \kh{be} a cherry or \kh{a} reticulated cherry of~$N$ that is not a \scr-cherry. If~$N(x,y)$ is strongly orchard then~$N$ is strongly orchard.
\end{lemma}
\begin{proof}
Assume that $N(x,y)$ is strongly orchard. Let~$D$ be any rooting of~$N$. 
\kh{We distinguish between the cases that $(x,y)$ is a cherry or reticulated cherry in $D$ or that this is not the case.}

First suppose that~$(x,y)$ is a cherry or reticulated cherry in~$D$. Then~$D(x,y)$ is a rooting of~$N(x,y)$ since $(x,y)$ is not an \scr-cherry. Since $N(x,y)$ is strongly orchard, $D(x,y)$ must be orchard. Hence,~$D$ is orchard. So~$N$ is strongly orchard. 

Now suppose that~$(x,y)$ is \kh{neither} a cherry \kh{nor a} reticulated cherry of~$D$. This can only happen if~$(x,y)$ is a cherry of~$N$ that is in a source component of~$N$ and the \kh{single} root of~$D$ subdivides an edge on the path between~$x$ and~$y$ in~$N$. Then modify~$D$ to a rooting~$D'$ by making the internal vertex of this path the root. Then we can use the argument from the previous paragraph \kh{to show that
$N$ is strongly orchard.}
\end{proof}

\begin{remark}
\add{It follows from Theorem~\ref{thm:orchard}, Lemmas~\ref{lem:remainsstronglyorchard} and \ref{lem:remainsnotstronglyorchard}, and Figure~\ref{fig:weak_strong_orchard_sequence} that there is a key distinction between the weakly and strongly orchard properties concerning the reduction of (reticulated) cherries in arbitrary order. Specifically, the property of a multi-semi-directed network being weakly orchard is preserved under arbitrarily reducing both cherries and reticulated cherries. In contrast, a semi-directed network being strongly orchard is preserved only under arbitrarily reducing cherries and reticulated cherries that are not scr-cherries.}
\end{remark}

\kh{We call a cherry picking sequence $s=(s_1,\ldots ,s_k)$ of $N_0=N$} \emph{strong} if, for each~$i\in\{1,\ldots ,k\}$, 
\leomay{it holds that if~$N_{i-1}$ has at least one \scr-cherry, then~$s_i$ is an \scr-cherry of~$N_{i-1}$. Note that if~$N_{i-1}$ has no \scr-cherries then, by definition of a cherry picking sequence,~$s_i$ is a cherry or reticulated cherry of~$N_{i-1}$. We now state
the second main result of this section.}

%\kh{it} holds that if~$s_i$ is an \scr-cherry of $N_{i-1}=N\circ (s_1,\ldots ,s_{i-1})$, then~$N_{i-1}$ has no cherries and no reticulated cherries that are not \scr-cherries, \kh{we obtain our other main result of this section.}
%\todo{I am wondering if we cannot refer back to just above Observation 1 for $N_{i-1}=N\circ (s_1,\ldots ,s_{i-1})$ as this would remove repetitiveness below(?) }

\begin{theorem}\label{thm:strongorchard}
Let~$N$ be a semi-directed network~$N$ on~$X$. Then~$N$ is strongly orchard if and only if, for each strong cherry picking sequence~$s=(s_1,\ldots ,s_k)$ of~$N$ and for each~$i\in\{1,\ldots ,k\}$, it holds that if~$s_i$ is an \scr-cherry of $N_{i-1}=N\circ (s_1,\ldots ,s_{i-1})$ then~$N_{i-1}$ has at least two \scr-cherries \kh{where we put $N_0=N$}.
\end{theorem}
\begin{proof}
First suppose that~$N$ is strongly orchard \kh{and assume for contradiction that} there exists a strong cherry picking sequence~$s=(s_1,\ldots ,s_k)$ of~$N$ such that \kh{there exists some $1\leq i\leq k$ such that} $s_i$ is a \scr-cherry of $N_{i-1}=N\circ (s_1,\ldots ,s_{i-1})$ and~$N_{i-1}$ has no other \scr-cherries. Since~$s$ is strong, $N_{i-1}$ has no cherries or reticulated cherries apart from~$s_i$. Since~$s_i=(x,y)$ is an \scr-cherry, there exists a rooting~$D_{i-1}$ of~$N_{i-1}$ where the root subdivides the edge of~$N_{i-1}$ incident to~$x$. Consequently, $D_{i-1}$ has no cherries or reticulated cherries and is therefore not orchard. Hence, $N_{i-1}$ is not strongly orchard and, by Lemma~\ref{lem:remainsstronglyorchard}, $N$ is not strongly orchard, a contradiction.

The other direction of the proof is by induction on~$k=|X|+|R|-2$, with~$R$ the set of reticulations of~$N$. For~$k=0$ the statement is trivially true. Assume~$k\geq 1$ and that for each strong cherry picking sequence~$s=(s_1,\ldots ,s_k)$ of~$N_0=N$ and for each~$i\in\{1,\ldots ,k\}$, it holds that if~$s_i$ is an \scr-cherry of $N_{i-1}=N\circ (s_1,\ldots ,s_{i-1})$ then~$N_{i-1}$ has at least two \scr-cherries.

First suppose that~$N$ has a cherry or a reticulated cherry~$(x,y)$ that is not an \scr-cherry. Since~$N(x,y)$ is strongly orchard by induction, it follows that~$N$ is strongly orchard by Lemma~\ref{lem:remainsnotstronglyorchard}.

Now suppose that~$N$ has no cherries or reticulated cherries that are not \scr-cherries. Then~$N$ has at least two \scr-cherries $(x,y)$ and~$(w,z)$ (possibly, $z=y$). Consider any rooting~\kh{$D$} of~$N$. Then at least one of~$(x,y)$ and~$(w,z)$ is a reticulated cherry in~\kh{$D$}. Assume without loss of generality that~$(x,y)$ is a reticulated cherry of~\kh{$D$}. Since~$N(x,y)$ is strongly orchard by induction, each rooting of~$N (x,y)$ is orchard. In particular, $D (x,y)$ is orchard. Hence,~\kh{$D$} is orchard. Since~\kh{$D$} was arbitrary, it follows that~$N$ is strongly orchard.
\end{proof}

\kh{The example in} Figure~\ref{fig:weak_strong_orchard_sequence} shows why in the characterization in Theorem~\ref{thm:strongorchard} ``for each'' cannot be replaced by ``there exists''. Even though there exists a cherry picking sequence of the required type in the depicted semi-directed network~$N$, $N$ is not strongly orchard. Indeed, $N$ also has a strong cherry picking sequence that is not of the required type.

\begin{figure}[htb]
    \centering
    \includegraphics[width=\textwidth]{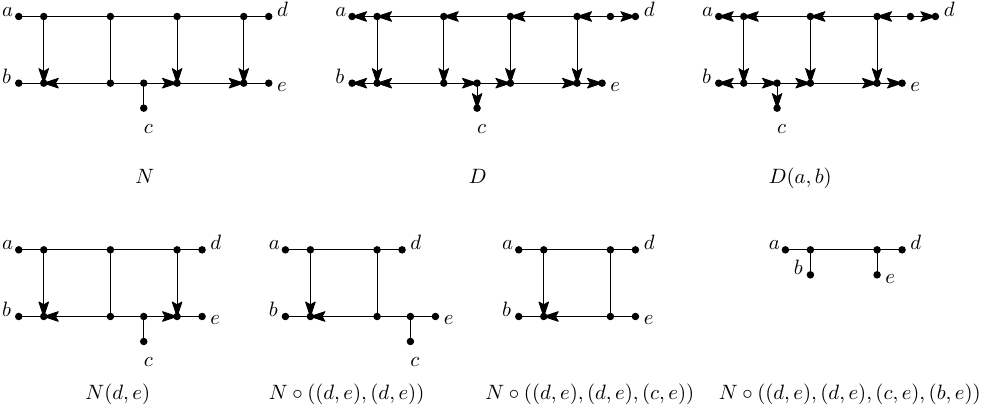}
    \caption{A semi-directed network~$N$ on \kh{$\{a,b,c,d,e\}$} that is weakly orchard but not strongly orchard, along with a rooting~$D$ that is not orchard. This is easy to see since \kh{$D$ has no cherry,} $(a,b)$ is the only reticulated cherry of~$D$, and the depicted network~$D (a,b)$ has no cherries or reticulated cherries. The bottom row of networks shows the first~$4$ networks obtain when reducing~$N$ by the strong cherry picking sequence \kh{$((d,e),(d,e), (c,e), (b,e), (a,b) ,(b,d))$} of~$N$. %which has the property that whenever an \scr-cherry is picked there is at least one other \scr-cherry 
    }
    \label{fig:weak_strong_orchard_sequence}
\end{figure}
%\todo{\kh{In the first network in the second row please replace "$N(d,e)$" by $N((d,e))$". NH: We already defined $N(d,e)$ as $N((d,e))$ for a single cherry, so I don't think this is necessary. Unless we think it make it clearer. (Then we might also want to change $D(a,b)$ to $D((a,b))$.)}}
% Leo: I don't think this would make it clearer. Better use the simplified notation if we define it

\section{Path Partitions} \label{sec:paths}

\vm{We now turn to our final concept for characterizing classes of (multi-)semi-directed networks networks called path partitions\add{, see also \cite{francis2018new,lafond2025path,huber2022forest}}.
Basically speaking, for a multi-rooted network, this is a partition 
of the vertex set of the network whose parts induce a collection of directed paths each of which must end in a leaf.
Multi-rooted networks that enjoy this property are called \emph{forest-based} networks and can be used to model
introgression \cite{scholz2019osf}; they are considered in more depth in \cite{lafond2025path,huber2022forest}.
In this section, \leojuly{we extend} the
theory of path partitions from 
multi-rooted networks to (multi-)semi-directed networks \leojuly{and} shall see that having a 
path partition characterizes weakly tree-based 
semi-directed networks (Corollary~\ref{cor:weak-tree-based}).}

%\kh{To motivate this concept consider the network \nh{$N_3$} in Figure~\ref{fig:not-weakly-tree-but-weakly-forest}. This is a multi-semi-directed networks that is weakly tree-based but not weakly tree-child. Despite this, $N_3$ is still interesting from a combinatorial point of view as it admits a particularly attractive spanning forest $F$ in the sense that the leaf set of $F$ is the same as the leaf set of $N_3$ and every arc removed from $N_3$ to obtain $F$ has its tail in one tree of $F$  and its head in another. Rooted networks that enjoy this property have been studied in the context of introgression \cite{scholz2019osf}, an important evolutionary process affecting both plants and animals. To be able to characterize multi-semi-directed networks that admit such a forest (Theorem~\ref{thm:weaklyforestbased}), we require further definitions.

%
\begin{figure}%[htb]
    \centering
    \includegraphics{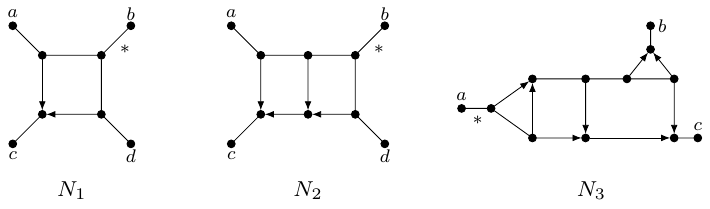}
    \caption{\label{fig:not-weakly-tree-but-weakly-forest}
    \nh{The semi-directed network~$N_1$ on $X = \{a,b,c,d\}$ is weakly tree-child, weakly forest-based and weakly tree-based. \kh{In each case, the rooting is obtained by subdividing the edge labelled $*$ to obtain the root.} The semi-directed network~$N_2$ on $X$ is weakly forest-based and weakly tree-based (since the rooting with the root subdividing the edge labelled $*$ is forest-based and tree-based), but not weakly tree-child. The semi-directed network $N_3$ on $\{ a,b,c\}$ is weakly tree-based (since the rooting with the root subdividing the edge labelled $*$ is tree-based), but not weakly tree-child and not weakly forest-based. }}
\end{figure}

 \kh{We begin with some
 definitions.} Suppose that $G$ is a mixed graph. 
 We call $G$ a \emph{forest} if it does not contain a cycle. If $D$ is a multi-rooted network on $X$ then we call \nh{$D$} \emph{forest-based} 
 if there exists a subgraph $F$ of \nh{$D$} in the form of a forest such that $F$ spans $V(D)$ and has leaf set $X$ and, 
 for each arc~$(u,v)\in A(D)\setminus A(F)$, the vertices $u$ and~$v$ are in different trees of~$F$. In this case, we call $F$ %the
 \nh{a} {\em support forest} of \nh{$D$}. Note that an element $T\in F$ might be a single 
 vertex or $T$ might contain a vertex $v$ such that \nh{$d_T(v)=2$}. If \nh{$N$} 
 is a multi-semi-directed network on~$X$, then we say that \nh{$N$ is} \emph{weakly forest-based} 
 if~\nh{$N$} has a rooting \kh{$D$} that is forest-based. 
 %\nh{with a support forest~$F$} that has leaf set~$X$.
 %\todo{to make sure that you can't root at a leaf so that it is no longer a leaf} and is forest-based. 
 \kh{Note that since the leaf sets of $D$ and \nh{$N$} coincide, no leaf of $N$ can be a root in $D$.} 
 If every rooting of \nh{$N$} with leaf set $X$ is forest-based, then we call \nh{$N$} \emph{strongly forest-based}. 
 \vm{See Figure~\ref{fig:forest-based} for some examples to illustrate these definitions.}
 \vm{Note if $N$ is semi-directed, then by \cite[Theorem 1]{huber2022forest} it follows that if $N$ is weakly (resp. strongly) forest-based then it is weakly (resp. strongly) tree-based, but not conversely.} %\todo{NH: I think the current version of this statement is now correct. See explanation in latex. For msdn's I think the whole comparison does not make sense anyway, since no 'real' MSDN (i.e. with more $>1$ roots) is ever WTB}%I think this is only true for semi-directed networks. Then, we can say that if $N$ is wfb, it has a fb rooting, which is also tb (by \cite{huber2022forest}, fig3), so $N$ is wtb. However, if N is a multi-semi-directed network that is wfb, this proof does not work, since \cite{huber2022forest} is for 1-rooted only. Moreover, if $N$ is a wfb msdn, it can never be wtb, since no rooting can have a rooted spanning tree (because there are multiple roots to be chosen in any rooting). (We also say after cor.3 that msdn's that are wtb are sdn's. 
 %So, for MSDN's I guess 

 %\todo{\kh{kathi says: We need to make sure that we really mean the support forest in the next 2 results and not the base forest!}}

\begin{figure}[htb]
    \centering
    \includegraphics[]{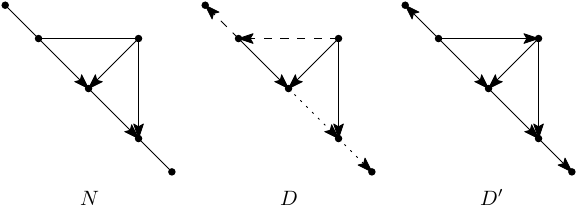}
    \caption{\label{fig:forest-based}A multi-semi-directed network~$N$ that is weakly forest-based but not strongly forest-based, along with a rooting~$D$ \kh{of $N$} that is forest-based (the dashed and dotted \kh{paths are the trees that make up the support} forest) and a rooting~$D'$ \kh{of $N$} that is not forest-based.}
\end{figure}

Now, suppose that $N$ is a multi-semi-directed network on $X$. 
\vm{For $x \in X$, let $P$ be either a semi-directed path $(u_1=u,u_2,\ldots, u_k=x)$, $k\geq 2$, in $N$
joining a vertex $u\in V(N)-X$ to $x$ or the \nhh{trivial} path $(u=x)$. Then we call $u$ the {\em handy vertex} of $P$. In addition,} for
a collection $\cP$ of \nhh{semi-directed} paths \kh{in $N$}, we call a maximal connected subgraph 
containing only edges $\{u,v\}$ of $N$ with $u,v$ in different \nhh{semi-directed} paths of~$\cP$ a \emph{cross component} of~$\cP$. \vm{We now show that a special type of path partition arises from 
weakly forest-based multi-semi-directed networks.}

%Let $P=(u_1=u,u_2,\ldots, u_k=v)$, some $k\geq 2$, denote a semi-directed path in $N$ joining a vertex $u\in V(N)-X$ with a vertex $v$ of $N$. Then we call $v$ an {\em end vertex of $P$} if there exists some $i\in\{1,\ldots, k-1\}$ such that $(u_i,u_{i+1})$ is an arc of $P$ or $P$ contains no arcs and $v$ is an element in $X$. for}\todo[inline]{NH: I don't think this is entirely what we want. We want the end-vertices to be the vertices $v_1$, $v_4$, $v_6$, $v_9$ and the leaves $a,b,g$ in the figure. (So, I think it would be better to call them `start vertices' or `top vertices')). I think the definition should be: ``Let $P=(u_1=u,u_2,\ldots, u_k=v)$, some $k\geq 2$, denote a semi-directed path in $N$ joining a vertex $u\in V(N)-X$ with a vertex $v$ of $N$. Then we call $u$ an {\em end vertex of $P$}. We call $x \in X$ an {\em end vertex of $P$} if instead $P$ is the trivial path $(x)$.''}

\begin{figure}[htb]
    \centering
\includegraphics{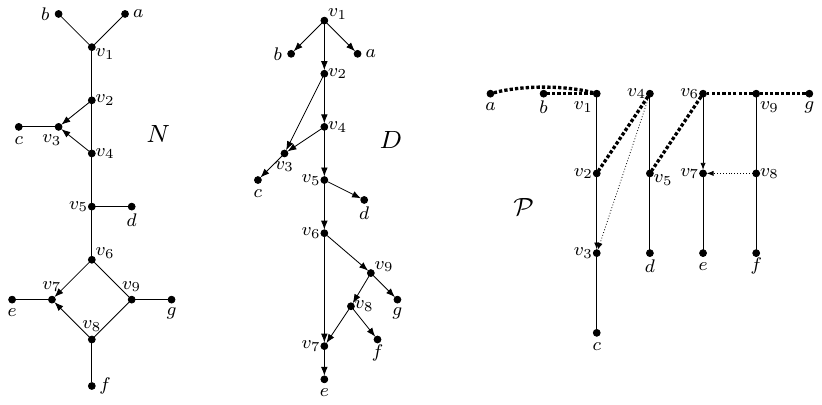}   
    \caption{A semi-directed network $N$ on $X = \{a, \ldots, g\}$ that is weakly forest-based \kh{as $D$ is a} rooting of $N$ that is forest-based. 
    \kh{The trivial paths $(a)$, $(b)$ and $(g)$ along with the four semi-directed paths made up by the thin black edges and arcs form a} collection of \nhh{semi-directed} paths~$\cP$ in~$N$ that satisfies Properties~(P1)-(P3) of Lemma~\ref{lem:pathsystems_conditions}. \kh{The vertices} \nh{$a$, $b$, $v_1$, $v_4$, $v_6$, $v_9$ and $g$ are the handy vertices} of the paths that contain them. The dotted thick black edges make up the three cross components of~$\cP$, while the thin black, dotted arcs are neither in a \nhh{semi-directed} path of $\cP$ nor in a cross component.}
    \label{fig:path-system}
\end{figure}

\begin{lemma}\label{lem:pathsystems_conditions}
\kh{Let $N$ be a multi-semi-directed network $N$ on $X$ that has a rooting $D$ with leaf set $X$ that is forest-based. 
%and let $D$ be a rooting of $N$. If $F$ is a forest spanning $V(D)$ such that $L(F) = X$, 
Then there exists a collection $\cP$ of \nhh{semi-directed} paths in $N$ such that}
\begin{enumerate}
    \item[(P1)] each vertex of~$N$ is in exactly one \nhh{semi-directed} path in~$\cP$;
    \item[(P2)] each \nhh{semi-directed} path in~$\cP$ is either a trivial path~$(x)$ with~$x\in X$ or a semi-directed path from \kh{a vertex in $V(N)-X$} to some element in $X$;
    \item[(P3)] each cross component of~$\cP$ contains at most one vertex that is not  \nh{the handy vertex}
    %\todo{\kh{We need to make clear that we mean the first vertex of a path and not the last one. Since the paths in $\cP$ are not directed this is difficult to say}} 
    of a \nhh{semi-directed} path in~$\cP$.
    %\item[XXX] if $\{u,v\}$ is an edge of~$N$ and $u,v$ are in different paths of $\cP$, then $u$ or $v$ is an endpoint of a path in $\cP$;
    %\item[XXX] there is at most one edge $\{u,v\}$ of $N$ between every pair of paths in~$\cP$;
    %\item[XXX] for every cross component $C$ of $\cP$, there exists at most one path in $\cP$ that can contain some vertex in an edge in $C$ that is not an endpoint of a path in $\cP$.
\end{enumerate}
Moreover, \kh{if $F$ is a support forest for $D$ then,} for each arc $(u,v) \in A(D) \setminus A(F)$, \kh{we have that} 
\begin{enumerate}
    \item[(P4)] if $\{u,v\}$ is an edge or $(u,v)$ is an arc of~$N$ and~$u,v\in P\in\cP$ then~$u$ and~$v$ appear consecutive on~$P$.
\end{enumerate}
\end{lemma}
\begin{proof}
\kh{Suppose that $F$ is  a support forest for $D$.} We can decompose~$F$ into a set \kh{$\mathcal S$ of (directed)} paths by, for each vertex of~$F$ with outdegree greater than~$1$ in~$F$, arbitrarily deleting all but one outgoing arc from~$F$. \kh{Let $\cP$ denote the collection of semi-directed paths in~$N$ that corresponds to~$\mathcal S$. 
%the collection of directed paths in~$D$ obtained this way. 
Note that the first arc of a path in $\mathcal S$ might be different from the first arc of the corresponding \nhh{semi-directed} path in $\cP$.} We show that Properties~(P1)-(P3) are satisfied by~$\cP$.

Clearly,~$\cP$ satisfies Properties~(P1) and~(P2). \kh{To see that $\cP$ satisfies Property~(P3), we claim first that for every edge $e$ in a cross component of $\cP$ one of the vertices incident with $e$ must be the \nh{handy vertex} of a \nhh{semi-directed} path in $\cP$. To see this, assume for contradiction that $e$ is an edge in a cross component of $\cP$ and none of the vertices incident with $e$, call them $u$ and $v$, is the \nh{handy vertex} of a \nhh{semi-directed} path in $\cP$. Then $u$ and $v$ must be an interior vertex of \kh{the} \nhh{semi-directed} path 
in $\cP$ that contains it, respectively. Since $F$ is a forest spanning $V(D)$, it follows that one of $u$ and $v$ must be a reticulation in $N$. Hence $(u,v)$ or $(v,u)$ must be an arc in $N$; a contradiction as $\{u,v\}$ is an edge in $N$. Hence, the claim holds}

%Next, we show that for every edge~$\{u,v\}$ in a cross component, $u$ or $v$ is an endpoint of some path $P \in \cP$. Specifically, it is not possible that~$u,v$ are internal vertices of two paths~$P_i,P_j$, because then one of~$u,v$ would be a reticulation and hence~$(u,v)$ or~$(v,u)$ would be an arc of~$N$ rather than $\{u,v\}$ an edge. Therefore, at least one of~$u,v$ is an endpoint of some~$P\in\cP$ as required.

\kh{Assume for} contradiction,  that Property~(P3) does not hold. Then 
\kh{there exists a cross component~$C$ of $\cP$ that contains two or more vertices that are not the \nh{handy vertex} of the \nhh{semi-directed} path in $\cP$ that contains them. Let $u$ and $u'$ denote two such vertices of $C$ and let $v$ and $v'$ be the vertices of $N$ such that $\{u,v\}$ and $\{u', v'\}$ are edges in $C$  and $u$ and $v$ are in different \nhh{semi-directed} paths of $\cP$ and $u'$ and $v'$ are in different \nhh{semi-directed} paths of $\cP$. Note that the \nhh{semi-directed} paths in $\cP$ that contain $u$ and $u'$ might be the same.} Then, by the previous \kh{claim}, $v$ and  $v'$ \kh{must be the \nh{handy vertices}} of \nhh{semi-directed} paths in~$\cP$. \kh{Since, by Theorem~\ref{thm:semi-directed-cycle}, $N$ cannot contain a semi-directed cycle, %because it is a multi-semi-directed network, 
it follows  that
$C$ is} a tree. Therefore, \kh{there exists a path $U$ in $C$ joining $u$ and $u'$. Replacing, if necessary, $v$ with the vertex on $U$ adjacent with $u$ and $v'$ with the vertex on $U$ adjacent with $u'$,
it follows that $U$ has the form $(u, v, \ldots, v', u')$}. 
%Otherwise, by the previous \kh{claim} and since $u \neq u'$, there will be some other pair of edges in this path~$U$ satisfying this property.

%\kh{Assume for} contradiction,  that Property~(P3) does not hold. Then, \kh{there exists a cross component $C$ of $\cP$ and vertices $u,u',v,v'\in V(D)$ such that}  $\{u,v\}$ and $\{u', v'\}$ are edges in $C$ (with possibly $v = v'$ \kh{holding}) such that $u \neq u'$ and \kh{$u$ and $u'$ are} not end vertices of \kh{the respective paths $P_u$ and $P_{u'}$} in $\cP$ \kh{that contain them. Note that $P_u=P_{u'}$ might hold.} By the previous \kh{claim}, $v$ and  $v'$ have to be \kh{the end vertices} of paths in~$\cP$. \kh{Since, by Theorem~\ref{thm:semi-directed-cycle}, $N$ cannot contain a semi-directed cycle because it is a multi-semi-directed network, it follows  that$C$ is} a tree. Therefore, \kh{there exists a path $U$ in $C$ joining $u$ and $u'$.} \kh{Swapping the roles of $u$ and $u'$ with those of  $v$ and $v'$ if necessary,} we may assume without loss of generality that $U$ has the form $(u, v, \ldots, v', u')$ \kh{where again $v=v'$ might hold}. 
%Otherwise, by the previous \kh{claim} and since $u \neq u'$, there will be some other pair of edges in this path~$U$ satisfying this property.

\kh{Since both} $\{u,v\}$ and $\{u', v'\}$ \kh{are edges of $N$ and so} must be oriented in any rooting of $N$ \kh{and, in combination, Properties~(P1) and (P2) imply that any arc on a path $P$ in $\cP$ is oriented towards the \nh{handy vertex} of $P$}, it follows that $\{u,v\}$ and $\{u', v'\}$ must be oriented towards each other in \kh{the orientation of $C$ induced by $D$. If $v\not=v'$ this is} not possible since \kh{it implies that} some edges in $U$ must have been arcs in~$N$. \kh{If $v=v'$ then $v$ is a reticulation in $N$. Hence, $N$ contains the arcs $(u,v)$ and $(u',v)$ as it is a multi-semi-directed network; a contradiction since, by assumption, $\{u,v\} $ and $\{u',v\}$ are edges in $C$ and therefore in $N$.}

\kh{The remainder is an immediate consequence of the fact that $F$ is a support forest for $D$.}
\end{proof}

\nh{Using \kh{Lemma~\ref{lem:pathsystems_conditions}}, we can now characterize multi-semi-directed networks that are weakly forest-based}
\vm{(note that this is an analogue of \cite[Theorem 1]{huber2022forest}).}

\begin{theorem}\label{thm:weaklyforestbased}
A multi-semi-directed network~$N$ on~$X$ is weakly forest-based if and only if there exists a collection~$\cP$ of \nhh{semi-directed} paths in~$N$ satisfying Properties~(P1)-(P4) of Lemma~\ref{lem:pathsystems_conditions}.
\end{theorem}
\begin{proof}
    %The `only-if' direction follows directly from the definition of weakly forest-based and Lemma~\ref{lem:pathsystems_conditions}.

    %For the `if' direction, suppose 
\kh{By Lemma~\ref{lem:pathsystems_conditions}, it suffices to show that if $N$ has a collection of \nhh{semi-directed} paths that satisfies Properties~(P1)-(P4) then $N$ must be weakly forest-based. So suppose}
    that $\cP$ is a collection of \nhh{semi-directed} paths in $N$ that satisfies \kh{these properties.}
    %Properties~(P1)-(P4) of Lemma~\ref{lem:pathsystems_conditions}. 
    We \kh{start with  associating a directed graph $D$ to $N$ and then show that $D$ is  in fact a} rooting of~$N$ that is forest-based.

    \kh{To obtain $D$, we employ Property~(P2) and orient, for each non-trivial \nhh{semi-directed} path $P\in\cP$,} all edges \kh{on $P$} towards \kh{the unique element in $X$ it contains}. 
    %Note that this is possible since these paths are semi-directed (Condition~(P2)). 
    For each cross component $C$ in $\cP$, we do the following. Let $v_C$ be the unique vertex in~$C$ that is not a \nh{handy vertex} of a \nhh{semi-directed} path in~$\cP$ (Property~(P3)), if it exists. Otherwise, choose an arbitrary \kh{non-leaf} vertex~$v$ in~$C$ \kh{to play the role of $v_C$.}
    %\todo{NH: In the case where $v_C$ does not exist, I think we need to make sure that we pick a non-leaf vertex to play the role of $v_C$, otherwise this leaf becomes the root (see figure~12, where we could pick $a$ as $v_C$ in $C_1$). This is fine, since every cross component must contain a non-leaf.}. 
    \kh{In either case, we} then orient all edges in~$C$ away from~$v_C$. Observe that this is well-defined since $C$ is a tree by Theorem~\ref{thm:semi-directed-cycle}. The obtained directed graph \kh{is $D$}.

    We now show that~$D$ is a rooting of~$N$. \kh{To this end, we need to show that $D$ is a multi-rooted network such that $N$ is a semi-deorientation of $D$. Clearly, $D$ is a multi-rooted network since it cannot} contain a directed cycle \kh{as} otherwise~$N$ would have contained a semi-directed cycle, which is not allowed by Theorem~\ref{thm:semi-directed-cycle}.  \kh{To see that $N$ is a semi-deorientation of $N$, we need to show that the reticulations of $D$ are precisely the reticulations of $N$. Since every reticulation of $N$ is also a reticulation of $D$, it suffices to show that every reticulation of $D$ is also a reticulation of $N$.}
    \kh{Assume for contradiction} that~$D$ contains a reticulation~$r$ that is \emph{not} a reticulation of~$N$. \kh{Then since neither a leaf nor a root of $N$ can be a reticulation in $D$, it follows that} $d^e_N (r) = d_N (r) \geq 3$. Hence, at least one of the edges incident to $r$ in $N$ is in a cross component~$C$ of $\cP$, and so $r$ is \kh{a vertex} in~$C$. \kh{By} Property~(P3), \kh{it follows} that $r$ is a \nh{handy vertex} of some \nhh{semi-directed} path in $\cP$ or $r = v_C$ \kh{holds. Hence,} $C$ is oriented \kh{in $D$} in such a way that $r$ has exactly one incoming arc. \kh{Consequently}, $r$ cannot be a reticulation in $D$, a contradiction. \kh{Thus, $D$ must be a rooting of $N$.}
    
    %\todo[inline]{NH: I agree, the current version is a bit confusing. I think it should be something like [starting after the `Assume for contradiction' sentence]: ``Since $r$ is not a reticulation in~$N$, we have that $\delta^e_N (r) = \delta_N (r) \geq 3$. Hence, at least one of the edges incident to $r$ in $N$ is in a cross component~$C$ of $\cP$, and so $r$ is also in~$C$. Property~(P3) then implies that $r$ is an end vertex of a path in $\cP$ or $r = v_C$. Then, $C$ has been oriented in such a way that $r$ as exactly one incoming arc in~$D$. Hence, $r$ cannot be a reticulation in $D$, a contradiction.''}

    It remains to show that $D$ is forest-based. \kh{Let $\cP'$ denote the collections of \nhh{semi-directed} paths in $D$ induced by $\cP$. Then Property~(P1) implies that} $\cP'$ is a forest spanning $V(D)$. \kh{Furthermore, Property~(P2) implies that the}  leaf set \kh{of $D$ is} $X$. Lastly, Property~(P4) \kh{implies that}, for each arc $a \in A(D) \setminus A(\cP')$,  \kh{the head of $a$ and the tail of $a$} are in different \nhh{semi-directed} paths of $\cP'$. Hence, \kh{$\cP'$ is a support} forest for $D$. \kh{It follows that} $N$ is weakly forest-based.
\end{proof}

\kh{As it turns out, in multi-semi-directed networks that are forest-based, collections of \nhh{semi-directed} paths that satisfy Properties~(P1)-(P3) in Lemma~\ref{lem:pathsystems_conditions} also turn out to hold the key for
our characterization of semi-directed networks that are weakly tree-based.}

%\nh{Recall from Section~\ref{sec:omnians} that} a multi-rooted network on $X$ is called tree-based if it has a rooted spanning tree with leaf set $X$ and that a a semi-directed network $N$ is called weakly tree-based if there exists a rooting of $N$ that is tree-based. \nh{and characterize such semi-directed networks in the following corollary of Theorem~\ref{thm:weaklyforestbased}.}

\begin{corollary}\label{cor:weak-tree-based}
A semi-directed network~$N$ on~$X$ is weakly tree-based if and only if there exists a collection~$\cP$ of \nhh{semi-directed} paths in~$N$ satisfying Properties~(P1)-(P3) of Lemma~\ref{lem:pathsystems_conditions}.
\end{corollary}
\begin{proof}
    \kh{By the definition of weakly tree-based and Lemma~\ref{lem:pathsystems_conditions}, it suffices to show that if $N$ contains a collection of \nhh{semi-directed} paths that satisfy Properties~(P1)-(P3) of Lemma~\ref{lem:pathsystems_conditions} then
    $N$ is weakly tree-based.} Suppose that $\cP$ is a collection of \nhh{semi-directed} paths in $N$ \kh{that satisfies these properties}
    %satisfying Properties~(P1)-(P3) of Lemma~\ref{lem:pathsystems_conditions}.
    We create a rooting~$D$ of~$N$ as in the proof of Theorem~\ref{thm:weaklyforestbased}. Note that~$D$ now has a single root \kh{because $N$ is semi-directed}. To see that $D$ is tree-based, observe that $\cP$ again spans \kh{$V(D)$} (by Property~(P1)) and has leaf set $X$ (by Property~(P2)). By \cite[Theorem~2.1]{francis2018new}, this means that $D$ is tree-based. \kh{Hence,} $N$ is weakly tree-based.
\end{proof}

\section{Discussion} \label{sec:discussion}

In this paper, we have introduced new explicit mathematical characterizations 
of multi-semi-directed networks, overcoming the need to implicitly define 
such networks through their rooted counterparts. By extending existing concepts 
for rooted networks - such as cherry picking sequences, omnians and 
path partitions - we have been able to explicitly characterize when a multi-semi-directed 
network has a rooting that is within some commonly studied classes of rooted 
networks. In  particular, with the growing interest in semi-directed networks (see e.g. \cite{banos2019identifying,gross2018distinguishing,gross2021distinguishing,linz2023exploring,xu2023identifiability,holtgrefe2025squirrel}), our characterizations have the potential to 
make the mathematical analysis of the algebraic models associated with semi-directed networks more tractable 
(see e.g. \cite{maxfield2025dissimilarity,englander2025identifiability,allman2025level1identifiability}).

%We have focused on the most prominent classes of networks that lack clear 
%explicit characterizations. For some of our studied classes, 
Although we have presented some characterizations for when a multi-semi-directed
network is contained within a certain class, there remain some
open questions in this direction. For example, 
generalizing our characterization of strongly orchard semi-directed
networks (Theorem~\ref{thm:strongorchard}) to multi-semi-directed networks, 
deciding whether or not strongly forest-based networks can be characterized with path partitions, and
seeing if weakly orchard networks can be characterized in terms of an HGT-consistent 
labelling directly applied to the multi-semi-directed network are all interesting questions.
In addition, similar questions could be investigated for other well-known network classes, including
but not limited to \emph{proper-forest-based}, \emph{normal}, \emph{reticulation visible} 
and \emph{tree-sibling} networks \cite{kong2022classes}.

Finally, our results open up a number of interesting algorithmic questions. We 
have already sketched an efficient algorithm to check if a mixed graph is 
a (multi-)semi-directed network (see the end of Section~\ref{sec:rootings}). 
Developing efficient algorithms to check whether
a given (multi-)semi-directed network lies within a fixed class would be a logical next step. For some classes efficient  
algorithms follow directly from our results (see e.g. Theorem~\ref{thm:stronglyTC}). 
However, for other classes the existence of efficient algorithms is not immediately obvious. 
For example, it would be interesting to determine whether or not there exists a linear 
time algorithm to check if a semi-directed network is strongly orchard. This 
question can be answered affirmatively for rooted networks, but for semi-directed 
networks a naive algorithm takes quadratic time (by checking every possible rooting).

% \begin{itemize}
%     \item Is there a characterization of multi-semi-directed (binary or non-binary) orchard networks~$N$ in terms of HGT-consistent labellings directly applied to~$N$?
%     \item Can we characterize strongly forest-based?
%     \item Can the cherry picking proposition be generalized to multi-rooted? The idea would be that the alternative \scr-cherry must be in the same source component.
% \end{itemize}

% Algorithms to check characterizations:
% \begin{itemize}
%     \item Can someone find a linear-time algorithm to recognize if a given network is strongly orchard? Notice we can do it in quadratic time by checking for every rooting of the network whether it is orchard (which takes linear time).
%     \item Can we efficiently check whether a mixed graph is a (multi-)semi-directed network?
% \end{itemize}

% Open problems posed by Cecile Ane at ICERM:
% \begin{itemize}
%     \item Are there restrictions (e.g. root on edges only, not at degree-3 nodes) such that weakly tree-child implies strongly tree child?
%     \item Is there an appropriate notion of cherry-picking operations on semi-directed networks?
%     \item Characterize “orchard” semi-directed networks. (strongly vs. weakly orchard: all vs. one rooted partners are orchard).
%     \item Is there a cherry-picking characterization for a semi-directed network to be tree-child?
% \end{itemize}

\section*{Acknowledgements}
This paper is based on research that was partly carried out while the authors were in residence at the Institute for Computational and Experimental Research in Mathematics \add{(ICERM)} in Providence, RI, during the semester program on ``Theory, Methods, and Applications of Quantitative Phylogenomics''. We thank the organizers for organizing this semester program. We thank C\'ecile An\'e for posing an open question during this semester program regarding weakly and strongly tree-child semi-directed networks, which motivated part of this research. \add{We thank the reviewers for their valuable comments and suggestions.} %NH, LvI \& MJ were supported by the Dutch Research Council (NWO) under Grant No.~OCENW.M.21.306, and LvI \& MJ also under Grant No.~OCENW.KLEIN.125.

\bibliographystyle{plain}
\bibliography{bib}

\section*{Statements \& Declarations}

\subsection*{Funding}
This material is based upon work supported by the National Science Foundation (NSF) under Grant No.~DMS-1929284 while the authors were in residence at the Institute for Computational and Experimental Research in Mathematics \add{(ICERM)} in Providence, RI, during the semester program on
``Theory, Methods, and Applications of Quantitative Phylogenomics''. NH, LvI \& MJ were supported by the Dutch Research Council (NWO) under Grant No.~OCENW.M.21.306, and LvI \& MJ also under Grant No.~OCENW.KLEIN.125.

\subsection*{Competing Interests}

The authors have no relevant financial or non-financial interests to disclose.

\subsection*{Author Contributions}

All authors contributed equally.

\subsection*{Data Availability}

No data was used.

%“All authors contributed to the study conception and design. Material preparation, data collection and analysis were performed by [full name], [full name] and [full name]. The first draft of the manuscript was written by [full name] and all authors commented on previous versions of the manuscript. All authors read and approved the final manuscript.”

%Question: Is a multi-semi-directed network strongly forest-based if and only if it can be partitioned into paths, such that each path starts with an arc, and no path has a shortcut arc? Seems to be nonsense. What about trees?

\end{document}